\documentclass[10pt,twocolumn,twoside]{IEEEtran}
\usepackage{cite}
\usepackage{graphicx,color,overpic,psfrag}
\usepackage{subfigure}
\usepackage{amsmath}
\usepackage{times}
\usepackage{latexsym}
\usepackage{bm}
\usepackage{amssymb}
\usepackage{cases}
\usepackage{array}
\usepackage{fancyhdr}
\usepackage{stfloats}
\usepackage{setspace}
\usepackage{multirow}
\usepackage{slashbox}
\usepackage{psfrag}
\usepackage{color}

\usepackage{caption}
\usepackage{algorithm}
\usepackage{algorithmicx}
\usepackage{algpseudocode}

\newtheorem{prop}{Proposition}
\newtheorem{lemma}{Lemma}
\newtheorem{remark}{Remark}

\newtheorem{definition}{Definition}

\newtheorem{corollary}{Corollary}

\begin{document}
\title{Jamming-Aided Secure Communication \\
in Massive MIMO Rician Channels}

\author{\authorblockN{Jue Wang, {\em Member, IEEE,}~Jemin Lee, {\em Member, IEEE,}\\
~Fanggang Wang, {\em Member, IEEE,} and Tony Q. S. Quek, {\em Senior Member, IEEE}}\\
\thanks{Manuscript received December 19, 2014, revised April 21, 2015. The editor coordinating the review of this paper and
approving it for publication was Z. Han.}
\thanks{Jue Wang, Jemin Lee, and Tony Q. S. Quek are with Singapore University of Technology and Design, Singapore 487372. Email: \{jue\_wang,~jemin\_lee,~tonyquek\}@sutd.edu.sg}
\thanks{Jue Wang is also with School of Electronic and Information Engineering, Nantong University, Nantong 226019, China.}
\thanks{Fanggang Wang is with State Key Laboratory of Rail Traffic Control and Safety,
Beijing Jiaotong University, Beijing 100044, China. Email: wangfg@bjtu.edu.cn}
\thanks{The contact author is J. Lee.}
\thanks{This research was partly supported by the Temasek Research Fellowship, the National Natural Science Foundation 61401240, 61201201, U1334202, the Key Grant Project of Chinese Ministry of Education (No. 313006),  the Fundamental Research Funds for the Central Universities (2015JBM112), and the State Key Lab of Rail Traffic Control and Safety (No. RCS2014ZT09).}
\thanks{Part of this work has been presented in {\em IEEE Global Commun. Conf., Workshop on  Massive MIMO: From theory to practice}, Austin, TX, Dec. 2014.}}

\markboth{IEEE Transactions on Wireless Communications}{}

\maketitle

\begin{abstract}

In this paper, we investigate the artificial noise-aided jamming design for a transmitter equipped with large antenna array in Rician fading channels. We figure out that when the number of transmit antennas tends to infinity, whether the secrecy outage happens in a Rician channel depends on the geometric locations of eavesdroppers. In this light, we first define and analytically describe the secrecy outage region (SOR), indicating all possible locations of an eavesdropper that can cause secrecy outage. After that, the secrecy outage probability (SOP) is derived, and a jamming-beneficial range, i.e., the distance range of eavesdroppers which enables uniform jamming to reduce the SOP, is determined. Then, the optimal power allocation between messages and artificial noise is investigated for different scenarios. Furthermore, to use the jamming power more efficiently and further reduce the SOP,
we propose directional jamming that generates jamming signals at selected beams (mapped to physical angles) only, and power allocation algorithms are proposed
for the cases with and without the information of the suspicious area, i.e., possible locations of eavesdroppers. We further extend the discussions to multiuser and multi-cell scenarios. At last, numerical results validate our conclusions and show the effectiveness of our proposed jamming power allocation schemes.

\end{abstract}


\section{Introduction}

Massive multiple-input multiple-output (MIMO) systems, where an enormous number of antennas are deployed at the base station, have become a hot research area in recent years \cite{Rusek:SigMag13,Larsson:ComMag14}. As the number of antennas goes to infinity, the effect of uncorrelated interferences and noises can tend to zero asymptotically by using only simple linear transmit/receive techniques \cite{Marzetta:TWC10}, leading to intensive growth in spectrum and power efficiency \cite{Ngo:TCom13}. When used for beamforming, massive MIMO leads to sharp beam patterns as well as low power leakage to unintended directions \cite{Alrabadi:JSAC13}. Due to these attractive properties, massive MIMO becomes a promising technique for future communication systems such as the fifth generation cellular system \cite{Andrews:JSAC14}. In the meanwhile, it can be anticipated that massive MIMO will also become crucial in security related applications.

Secure communication in wiretap channels has been studied for decades since the seminal work \cite{Wyner}. Corresponding studies have been further extended to different type of wiretap channels \cite{Csiszar78,Ozarow85}, fading channels \cite{Li:ISIT07,Gopala:TIT08}, MIMO channels \cite{Shafiee:TIT09,Oggier,Khisti:ITp1,Khisti:ITp2}, and networks \cite{Lee:JSAC13,Win:Netw14}. The research topics span a wide range from information-theoretical contributions such as secrecy capacity analysis and rate region characterization to practical transmission design issues including precoding, user scheduling, and artificial noise (AN)-aided jamming. For a complete review of the most lately approaches, see \cite{Hong:SigMag13,Survey}.
Regarding the communication secrecy, the emergence of the massive MIMO technique brings new opportunities and challenges.
{Recently, physical layer security techniques using massive MIMO have drawn increasing attentions in the literature. In \cite{Geraci:TCom12,Geraci:JSAC13,Geraci:TWC14,Jun}, the secrecy rate in massive MIMO systems has been analyzed using large system analysis and secure precoding schemes were designed. In \cite{XiongTIFS15,PIMRC}, it has been shown that massive MIMO can benefit the detection of active eavesdropper who performs attacking on the channel training phase \cite{Zhou}. Note that the above-mentioned approaches require either the channel state information (CSI) of eavesdroppers can be known, or their existence can be detected. For the scenarios that eavesdroppers are completely passive and their CSI is unknown, AN-aided jamming \cite{Goel:TWC08} can be a feasible solution. Only recently, the AN-aided jamming approach has been applied for massive MIMO systems in \cite{Zhu:Arxiv14,Zhu:EW14} and was shown to be beneficial for communication secrecy.}


In this paper, we study the secure communication in massive MIMO systems via AN-aided jamming.
Differently from \cite{Zhu:Arxiv14}, we consider the scenario that eavesdroppers are randomly located around a legitimate transmitter equipped with large antenna array, and all channels follow Rician distribution. In this case, the geometric locations (described by both the {angle of arrival} and {distance to the transmitter}) of the legitimate receivers and eavesdroppers become essential in the secrecy outage analysis, which highlights the main difference between our work and \cite{Zhu:Arxiv14}.
The motivation of our paper is based on the following considerations: 1) Since the beam towards the legitimate receiver becomes sharper and the power leakage to other directions becomes trivial in massive MIMO systems, it is doubtful whether jamming is still beneficial for secrecy, and 2) as the number of antennas grows, the dimension of the jamming space increases and jamming power needs to spread over a large number of directions, which makes conventional {\em uniform} jamming inefficient with massive MIMO. Regarding these issues,
two questions are raised:
\begin{enumerate}
  \item Does conventional uniform jamming still benefit the secure communication in massive MIMO systems when $N_t$ goes to infinity?
  \item Is there more efficient scheme rather than uniform jamming in the massive MIMO setup?
\end{enumerate}
In this paper, we will answer these two questions by making the following contributions:
\begin{itemize}
  \item
  For the massive MIMO Rician fading channels, we analytically describe the secrecy outage region (SOR) as geometric locations of eavesdroppers that can induce secrecy outage. The concept of SOR further has been used to characterize the secrecy outage probability (SOP).
  \item
  With the information of the {\em suspicious area} where eavesdroppers are possibly located, we derive analytical expression of the SOP in the presence of one legitimate receiver and multiple passive eavesdroppers. After that, it is proved that conventional uniform jamming is still useful in terms of reducing the SOP when any eavesdroppers are located within a certain distance range to Alice, which we call it as the jamming-beneficial range. This conclusion provides an answer to the first question.
  \item
  For uniform jamming, the optimal signal and jamming power allocation is investigated for different scenarios. We further devise practical {\em directional jamming} algorithms, either with or without the information of the suspicious area. The proposed directional jamming schemes use the jamming power more efficiently to further and substantially reduce the SOP, which provides answers to the second raised question.
\end{itemize}

The rest of this paper is organized as follows: Section II provides system model. Section III describes the SOR, further provides an analytical expression of SOP and a jamming-beneficial range. Optimal jamming power allocation is studied for uniform jamming in Section IV, and in Section V, directional jamming algorithms are proposed. In Section VI, the SOR is discussed for multiuser and multi-cell scenarios. Section VII concludes this paper.

\section{System Model}

In this section, we first present the network model. As an important concept in subsequent analysis, we further define the normalized crosstalk between two wireless links and introduce its characteristics. Then, the AN-aided secure transmission and the definition of SOP are described.

\subsection{Network Model}

\begin{figure}[tbp]
\centering
  \includegraphics[width=0.9\columnwidth]{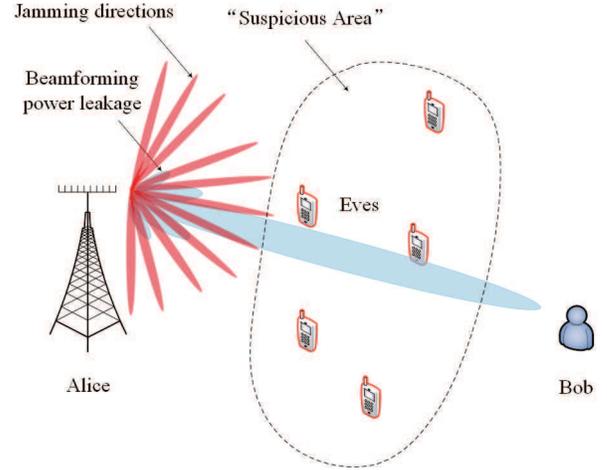}
  \caption{Description of the network layout.}
  \label{fig:SysMod}
\end{figure}

We consider the network shown in Fig.~\ref{fig:SysMod}, where a transmitter (Alice) equipped with $N_t$ antennas transmits to a single-antenna user (Bob) in the existence of $L$ external passive single-antenna eavesdroppers (Eves $1, ..., L$). Alice uses beamforming for the data transmission to Bob, while jamming with AN in other spaces (or directions). We define the set of receivers $\mathcal{I}_{\rm r} = \{{\rm b},{\rm e}_1, ..., {\rm e}_L\}$ where ${\rm b}$ denotes Bob and ${\rm e}_l$ $(l = 1,...,L)$ denotes Eve $l$. Considering Rician fading, the channel between Alice and receiver $i$ is given by
\begin{equation}
\label{eq:h leg}
  \mathbf{h}_i = \sqrt{\frac{K_i}{1 + K_i}}\bar{\mathbf{h}}_i + \sqrt{\frac{1}{1 + K_i}}\mathbf{g}_i, ~\forall i \in \mathcal{I}_{\rm r}
\end{equation}
where $K_i$ is the Rician $K$-factor, $\mathbf{g}_i \in \mathbb{C}^{N_t \times 1}$ is the i.i.d. fast fading part whose elements follow $\mathcal{CN}(0,1)$ distribution (complex normal distribution with zero mean and unit variance). For uniform linear array with inter-antenna spacing $d_0$ (in wavelength), the line of sight (LOS) component $\bar{\mathbf{h}}_i$ can be written as the steering vector at incident angle $\theta_i$:
\begin{equation}\label{eq:2}
  \bar{\mathbf{h}}_i = \bar{\mathbf{s}}(\theta_i) = \left(1, e^{-j2\pi d_0\sin \theta_i}, ..., e^{-j2\pi(N_t - 1)d_0\sin \theta_i}\right)^T
\end{equation}
where $\theta_i$ is the LOS angle of receiver $i$.
In addition, we consider large scale fading ${d_i^{-\alpha}}$
where $d_i$ is the distance from Alice to receiver $i$, and $\alpha$ is the path loss coefficient.

\newcounter{TempEqCnt}
\setcounter{TempEqCnt}{\value{equation}}
\setcounter{equation}{6}
\begin{figure*}
\begin{equation}\label{eq:19}
      F_{s_{i;j}}(x)
      =
      \begin{cases}
{1 - {F_{ \Delta }}\left({\rm CP}_0 (\frac{x}{K_{i;j}})\right),} ~~{{\rm{P}}{{\rm{V}}_1} \leq \frac{x}{K_{i;j}}  \leq 1 }\\
1 - {F_{ \Delta }}\left({\rm CP}_0 (\frac{x}{K_{i;j}})\right)
+ \sum\limits_{m = 1}^M {\left( {{F_{ \Delta }}\left({\rm CP}_{m,1}(\frac{x}{K_{i;j}})\right) - {F_{ \Delta }}\left({\rm CP}_{m,2}(\frac{x}{K_{i;j}})\right)} \right)},
~~{{\rm{P}}{{\rm{V}}_{M + 1}} \le \frac{x}{K_{i;j}}  < {\rm{P}}{{\rm{V}}_M}}
\end{cases}
\end{equation}
\hrulefill
\end{figure*}
\setcounter{equation}{\value{TempEqCnt}}

We consider a practical scenario that Eves are uniformly distributed within an angular range $\mathcal{A}_{\rm e} \triangleq [\theta_{\min}, \theta_{\max}]$ and a distance range $\mathcal{D}_{\rm e} \triangleq [{\rm D}_{\min}(\theta_{\rm e}), {\rm D}_{\max}(\theta_{\rm e})]$, where ${\rm D}_{\min}(\theta_{\rm e})$ and ${\rm D}_{\max}(\theta_{\rm e})$ are functions of $\theta_{\rm e} \in \mathcal{A}_{\rm e}$, defining two borders of this area. Throughout this paper, we use
\begin{equation}\label{eq:Rsus}
  \mathcal{R}_{\rm sus} \triangleq \{(\theta_{\rm e}, d_{\rm e})\left.\right| \theta_{\rm e} \in \mathcal{A}_e, d_{\rm e} \in \mathcal{D}_{\rm e}\}
\end{equation}
to define the {\em suspicious area}.
In practice, if Alice has only limited information of ${\rm D}_{\min}(\theta_{\rm e})$ and ${\rm D}_{\max}(\theta_{\rm e})$, she can assume the two boundaries are defined by constant values, $d_{\min}$ and $d_{\max}$.
For instance, if Alice knows nothing about $\mathcal{R}_{\rm sus}$, she can set $\mathcal{A}_{\rm e} = [0, 2\pi]$, $\mathcal{D}_{\rm e} = [0, r_{\rm max}]$, indicating that the suspicious area (from Alice's point of view) spans the entire space with radius $r_{\rm max}$.
The effectiveness of this assumption, referred to as ``constant boundaries'' and defined below, depends on that how accurately it can describe the real $\mathcal{R}_{\rm sus}$.

\begin{definition}[Constant Boundaries]
  To facilitate practical design, it is convenient to set the two boundaries of $\mathcal{R}_{\rm sus}$ to be constants such that
${\rm D}_{\rm min}(\theta_{\rm e}) = d_{\min},~{\rm D}_{\rm max}(\theta_{\rm e}) = d_{\max}, ~\forall \theta_{\rm e}\in\mathcal{A}_{\rm e}$.
\end{definition}

\subsection{Normalized Crosstalk}

For $\mathbf{h}_i$, $\mathbf{h}_j$ ($j \in \mathcal{I}_{\rm r}$) defined as \eqref{eq:h leg}, the following asymptotic results hold as $N_t \to \infty$ \cite{Zhang:Globecom13}:
    \begin{align}
    \label{eq:same channel large}
      \frac{1}{N_t} \mathbf{h}_i^H \mathbf{h}_i &\doteq 1\\
      \label{eq:diff channel large}
      \frac{1}{N_t} \mathbf{h}_i^H \mathbf{h}_j &\doteq \frac{1}{N_t}\sqrt{K_{i;j}} t_{i;j},~j\neq i.
    \end{align}
    where $\doteq$ denotes the approximation that is asymptotically accurate,\footnote{In this paper, we focus on the large antenna regime, and will use equalities instead of approximations for brevity.} $K_{i;j} \triangleq \frac{K_i K_j}{(1 + K_i)(1 + K_j)}$, and $t_{i;j} \triangleq \sum_{n = 0}^{N_t - 1}e^{-j2\pi d_0 (\sin\theta_i - \sin\theta_j) n}$. Stemming from \eqref{eq:diff channel large}, we introduce the following definition.
\begin{definition}[Normalized Crosstalk]
Define the normalized crosstalk between nodes $i,j$ as
  \begin{equation}\label{eq:23.2}
    s_{i;j}(\theta_i, \theta_j) \triangleq \left|\frac{1}{N_t} \mathbf{h}_{i}^H \mathbf{h}_j\right|^2
    \doteq \frac{1}{N_t^2} K_{{ i;j}} |t_{i;j}|^2.
  \end{equation}
\end{definition}

\begin{lemma}\label{le:CT cha}
  The normalized crosstalk $s_{i;j}(\theta_i, \theta_j)$ has the following characteristics:
  \begin{enumerate}
    \item $s_{i;j}(\theta_i, \theta_j)$ is a sinc-like function
    composed of one main lobe and multiple side lobes.
    \item With fixed $\theta_j$ and random $\theta_{i} \sim \mathcal{U}(\theta_{\rm min}, \theta_{\rm max})$ which is uniformly distributed between $\theta_{\rm min}$ and $\theta_{\rm max}$, the CDF of $s_{i;j}(\theta_i, \theta_j)$ can be written as \eqref{eq:19} (see top of this page)
  where the definitions of ${\rm CP}_0$, ${\rm CP}_{m,1}$, ${\rm CP}_{m,2}$ and $F_{\Delta}(\cdot)$ are referred to Appendix \ref{proof:le CT}.
\item With fixed $\theta_j$, the feasible range of $s_{i;j}(\theta_i, \theta_j)$ is $0 \leq s_{i;j} \leq s_{i;j}^{\rm max}$, where $s_{i;j}^{\rm max} \leq K_{i;j}$ is determined by the distribution range of $\theta_i$, i.e., $\mathcal{A}_{i} \triangleq [\theta_{{\rm min}}, \theta_{{\rm max}}]$.
  \end{enumerate}
\end{lemma}

\begin{proof}
  See Appendix \ref{proof:le CT}.
\end{proof}

\setcounter{equation}{7}

\subsection{Secrecy Transmission Scheme}
 We use linear precoding for data transmission, while AN symbols $s_n$ are sent in the space defined by $\mathbf{v}_n, n = 1,...,N$, to degrade the channels of Eves. For the null space-based jamming \cite{Goel:TWC08}, it holds that $N = N_t - 1$ and $\mathbf{v}_n \in {\rm null}(\mathbf{h}_{\rm b})$. The received signal at receiver $i$ is given by
\begin{equation}\label{eq:leg sig mod}
  y_i = \sqrt{P_{\rm b} d_i^{-\alpha}} \mathbf{h}_i^H \mathbf{w}_{\rm b} x_{\rm b} + \sum\limits_{n = 1}^N \sqrt{\bar{P}_n d_i^{-\alpha}}\mathbf{h}_i^H \mathbf{v}_n s_n + n_i,~\forall i \in \mathcal{R}
\end{equation}
where $\mathbf{w}_{\rm b} \in \mathbb{C}^{N_t \times 1}$ is the precoder for Bob, $x_{\rm b}$ is the unit-norm data symbol, and $n_i$ is the additive Gaussian noise. Moreover, $P_{\rm b}$ and $\bar{P}_n$ respectively are the powers allocated to Bob and the $n$-th jamming direction, with total power constraint such as $\sum_{n = 1}^N \bar{P}_n = P_{\rm tot} - P_{\rm b}$
where $P_{\rm tot}$ is the total available transmit power. We define the jamming power allocation coefficient as
\begin{equation}\label{eq:phi def}
  \phi \triangleq \frac{P_{\rm jam}}{P_{\rm tot}} = \frac{\sum_{n = 1}^N \bar{P}_n}{P_{\rm tot}}.
\end{equation}

For ease of description, we assume that Bob and all Eves share the same noise covariance being $N_0$. Moreover,
we consider maximum ratio transmission (MRT) for precoding of the data symbol $x_{\rm b}$, i.e.,
$\mathbf{w}_b = \frac{\mathbf{h}_b}{||\mathbf{h}_b||}$.
In this case,
according to \eqref{eq:leg sig mod}, the SINR at receiver $i$ is given by
\begin{equation}\label{eq:7}
  {\rm SINR}_i = \frac{P_{\rm b}d^{-\alpha}_i|\mathbf{h}_i^H \mathbf{w}_{\rm b}|^2}{N_0 + d^{-\alpha}_i \sum_{n = 1}^{N}\bar{P}_n|\mathbf{h}_i^H \mathbf{v}_n|^2}.
\end{equation}
We assume that the Eves are not colluding, but consider the most-capable Eve, which has the maximum receive SINR, to define the secrecy rate as \cite{Geraci:TWC14}
\begin{equation}\label{eq:sec rate}
  R_{\rm b}^{\rm s} = \left[\log_2(1 + {\rm SINR}_{\rm b}) - \log_2(1 + {\rm SINR}_{{\rm e},\max})\right]^+
\end{equation}
where ${\rm SINR}_{{\rm e},\max} \triangleq \max\limits_{l}{\rm SINR}_{{\rm e}_l}$ and $[x]^+ \triangleq \max\{x,0\}$. We say a secrecy outage occurs if $R_{\rm b}^{\rm s}$ is less than a target rate $R_{\rm th}$, hence the SOP is defined as
\begin{equation}\label{eq:SOP}
  {\mathcal P}_{\rm out} = \Pr\{R^{\rm s}_{\rm b} < R_{\rm th}\}.
\end{equation}

\section{Secrecy Outage Analysis}

In this section,
we first introduce the secrecy outage region (SOR) which describes all possible locations of Eves who can cause secrecy outage. Analytical expression of the SOR is derived for uniform jamming, then the SOP is studied with variant shapes of $\mathcal{R}_{\rm sus}$. At last, a jamming-beneficial range is derived to show that uniform jamming is still useful in reducing the SOP.

\setcounter{TempEqCnt}{\value{equation}}
\setcounter{equation}{21}
\begin{figure*}
\begin{equation}
  \label{eq:43.1}
  C_3(\phi) = \frac{C_2 (\phi)}{C_1 (\phi)} = \frac{\left({{1 + (1 - \phi)\tilde{P}_{\rm tot}d^{-\alpha}_b N_t} - 2^{R_{\rm th}}}\right)\tilde{P}_{\rm tot}\phi}{\left({{1 + (1 - \phi)\tilde{P}_{\rm tot}d^{-\alpha}_b N_t} - 2^{R_{\rm th}}}\right)\tilde{P}_{\rm tot}\phi + (1 - \phi)\tilde{P}_{\rm tot}N_t 2^{R_{\rm th}}}.
\end{equation}
\hrulefill
\end{figure*}
\setcounter{equation}{\value{TempEqCnt}}

\subsection{Secrecy Outage Region}
In the large antenna regime, all fast fading effects are completely averaged out as shown in \eqref{eq:same channel large} and \eqref{eq:diff channel large}. Therefore, whether the secrecy outage occurs or not, will be essentially determined by the geometric location of Eve.
In this light, we introduce the SOR defined in the following.
\begin{definition}[Secrecy Outage Region]\label{def:SOR}
  The SOR is defined in terms of polar coordinates as
  \begin{equation}\label{eq:def SOR}
    \mathcal{R}_{\rm SOR} \triangleq \left\{\left(\theta_{\rm e}, d_{\rm e}\right)\left.\right| \lim\limits_{N_t \to \infty}R_{\rm b}^{\rm s} < R_{\rm th}\right\}.
  \end{equation}
  Herein, we note that $R_{\rm b}^{\rm s}$ is a function of $\theta_{\rm e}$ and $d_{\rm e}$.
\end{definition}

In the large antenna regime, secrecy outage occurs if there exists at least one Eve within the SOR. If {\em all} Eves locate outside of the SOR, the target secrecy rate $R_{\rm th}$ can be guaranteed.
To characterize the SOR, we first evaluate the received SINRs assuming uniform jamming.

\begin{lemma}\label{le:SINR uj}
  With uniform jamming in ${\rm null}(\mathbf{h}_{\rm b})$, the SINRs at Bob and any Eve can be, respectively, written as\footnote{Hereafter, if one notation is applied for {\em any} Eve, we will use the subscript ${\rm e}$ instead of ${\rm e}_l$ for the sake of brevity.}
  \begin{align}
  \label{eq:29}
    {\rm SINR}_{\rm b}^{\rm uj} &= \tilde{P}_{\rm b} d^{-\alpha}_{\rm b} N_t\\
    \label{eq:30}
    {\rm SINR}_{{\rm e}}^{\rm uj} &=\frac{\tilde{P}_{\rm b}d^{-\alpha}_{{\rm e}}N_t s_{{\rm e};{\rm b}}(\theta_{\rm e})}{1 + d^{-\alpha}_{{\rm e}} \tilde{P}_{\rm jam} (1 - s_{{\rm e};{\rm b}}(\theta_{\rm e}))}
  \end{align}
where, and hereafter, we use the notations $\tilde{P}_{\rm b} = \frac{P_{\rm b}}{N_0}, ~\tilde{P}_{\rm jam} = \frac{P_{\rm jam}}{N_0}$ for brevity.
Note that in \eqref{eq:30}, $s_{\rm e;b}(\theta_{\rm e})$ is the normalized crosstalk between Eve and Bob as defined in \eqref{eq:23.2}.
Considering fixed $\theta_{\rm b}$, we hereafter write $s_{\rm e;b}$ as a function of only $\theta_{\rm e}$.
\end{lemma}

\begin{proof}
Since ${\rm span}(\mathbf{v}_n) = {\rm null}(\mathbf{h}_{\rm b})$, jamming causes no interference at Bob. Applying \eqref{eq:same channel large} to \eqref{eq:7}, we get \eqref{eq:29}.
  Noting that $N = N_t - 1$ and $\bar{P}_n = \frac{P_{\rm jam}}{N_t - 1}$, from \eqref{eq:7}, we have
  \begin{equation}\label{eq:31}
    {\rm SINR}_{{\rm e}}^{\rm uj} = \frac{P_{\rm b}d^{-\alpha}_{{\rm e}}|\mathbf{h}_{{\rm e}}^H \mathbf{w}_{\rm b}|^2}{N_0 + d^{-\alpha}_{{\rm e}} \frac{P_{\rm jam}}{N_t - 1} \sum\limits_{n = 1}^{N_t - 1}|\mathbf{h}_{{\rm e}}^H \mathbf{w}_n|^2}.
  \end{equation}
 By applying \eqref{eq:23.2}, the numerator of \eqref{eq:31} can be written as $P_{\rm b}d^{-\alpha}_{\rm e}N_t s_{\rm e;b}(\theta_{\rm e})$. On the other hand,
  noting that $\mathbf{w}_{\rm b}$ and $\mathbf{v}_n, n = 1,..., N_t - 1$ constitute a complete orthogonal basis of the $N_t$-dimensional vector space, we have $\sum_{n = 1}^{N_t - 1}|\mathbf{h}_{{\rm e}} \mathbf{w}_n|^2 {= ||\mathbf{h}_{{\rm e}}||^2 - |\mathbf{h}_{{\rm e}} \mathbf{w}_{\rm b}|^2} = N_t (1 - s_{{\rm e};{\rm b}}(\theta_{\rm e}))$ in the denominator of \eqref{eq:31}. Therefore, \eqref{eq:30} can be obtained.
\end{proof}

\begin{remark}
The result in \eqref{eq:29} leads to a constraint on $\phi$ (defined in \eqref{eq:phi def}), written as
\begin{equation}\label{eq:22.1}
  \phi \leq \phi^{\max} = 1 - \frac{2^{R_{\rm th}} - 1}{\tilde{P}_{\rm tot}d^{-\alpha}_{\rm b}N_t}
\end{equation}
which stems from the fact that the jamming power cannot be too large, otherwise, even without Eves, the target rate $R_{\rm th}$ cannot be guaranteed since the remained signaling power is too small. Unless otherwise specified, we assume \eqref{eq:22.1} can always hold via proper power allocation.
\end{remark}

From Lemma \ref{le:SINR uj}, we characterize the SOR for the uniform jamming as follows.
\begin{prop}\label{prop:sor}
With uniform jamming in ${\rm null}(\mathbf{h}_{\rm b})$ and given $\phi$, the SOR is described as
\begin{equation}\label{eq:38}
  \mathcal{R}_{\rm SOR}^{\rm uj}(\phi) = \Bigg\{(\theta_{\rm e}, d_{\rm e})\left.\right|d_{\rm e} < \bar{d}_{\rm e}^{\rm uj}(\phi, \theta_{\rm e}),
  s_{\rm e;b}\left(\theta_{\rm e}\right) > C_3(\phi)\Bigg\}
\end{equation}
where
\begin{align}\label{eq:39.1}
  \bar{d}_{\rm e}^{\rm uj}(\phi, \theta_{\rm e}) &= \left(C_1(\phi)s_{\rm e;b}\left(\theta_{\rm e}\right) - C_2(\phi)\right)^{\frac{1}{\alpha}}\\
\label{eq:40.1}
  C_1(\phi) &= \frac{(1 - \phi)\tilde{P}_{\rm tot}N_t 2^{R_{\rm th}}}{{1 + (1 - \phi)\tilde{P}_{\rm tot}d^{-\alpha}_b N_t} - 2^{R_{\rm th}}} +  \tilde{P}_{\rm tot} \phi,\\
  \label{eq:40.2}
  C_2 (\phi) &=  \tilde{P}_{\rm tot} \phi
\end{align}
and $C_3(\phi)$ is given by \eqref{eq:43.1} shown at the top of this page.
\end{prop}
\begin{proof}
  Substituting \eqref{eq:29} and \eqref{eq:30} into \eqref{eq:sec rate} and using the definition of SOR in \eqref{eq:def SOR}, we can obtain the value of $\bar{d}_{\rm e}^{\rm uj}(\phi, \theta_{\rm e})$ in \eqref{eq:39.1}. Note that the value of $\bar{d}_{\rm  e}^{\rm uj}(\phi, \theta_{\rm e})$ should be positive, this straightforwardly introduces the constraint on the minimum value of $s_{\rm e;b}(\theta_{\rm e})$ such as $s_{\rm e;b}\left(\theta_{\rm e}\right) > C_3(\phi)$ where $C_3(\phi)$ can be readily obtained by letting $C_1(\phi)s_{\rm e;b}(\theta_{\rm e}) - C_2(\phi) > 0$.
\end{proof}

\setcounter{equation}{22}

From Lemma \ref{le:CT cha}, $s_{\rm e;b}(\theta_{\rm e})$ in \eqref{eq:39.1} is a function  with one main lobe and multiple side lobes, resulting in a multi-lobe shaped SOR.
In order to gain some insights from this complex shape (as will be shown in the simulations), we focus on several critical security-related metrics as
\begin{itemize}
  \item {\em Largest radius of the main lobe ($\bar{d}_{{\rm e},0}^{\rm uj}$) and the $m$-th side lobe ($\bar{d}_{{\rm e},m}^{\rm uj}$):}
  $\bar{d}_{{\rm e},0}^{\rm uj}$ and $\bar{d}_{{\rm e},m}^{\rm uj}$ can be obtained by replacing $s_{\rm e;b}(\theta_{\rm e})$ in \eqref{eq:39.1} respectively with $K_{\rm e;b}$ and $K_{\rm e;b}{\rm PV}_{m}$ (defined as \eqref{eq:peak value} in the proof of Lemma \ref{le:CT cha}), such that
  \begin{align}\label{eq:43.2}
    \bar{d}_{{\rm e},0}^{\rm uj} &\triangleq \left(C_1(\phi)K_{\rm e;b} - C_2(\phi)\right)^{\frac{1}{\alpha}}\\
    \label{eq:radius side lobe}
    \bar{d}_{{\rm e},m}^{\rm uj} &\triangleq \left(C_1(\phi)K_{\rm e;b}{\rm PV}_{m} - C_2(\phi)\right)^{\frac{1}{\alpha}}.
  \end{align}
  Since $\bar{d}_{{\rm e},0}^{\rm uj}$ is much larger than $\bar{d}_{{\rm e},m}^{\rm uj}, \forall m \neq 0$, $\bar{d}_{{\rm e},0}^{\rm uj}$ can be considered as the {\em largest distance} of the SOR.
  For any Eve whose distance to Alice is larger than $\bar{d}_{{\rm e},0}^{\rm uj}$, we can conclude that it causes no secrecy outage regardless of its LOS direction.

  \item {\em Largest angle difference $\Delta \theta_{\max}$ of the SOR:} For any Eve whose angle difference to the LOS direction of Bob is larger than $\Delta \theta_{\max}$, we can conclude that it causes no secrecy outage regardless of its distance to Alice. If $s_{\rm e;b}(\theta_{\rm e}) < C_3(\phi), \forall \theta_{\rm e} > \hat{\theta}_{\rm e}$, we can write
      \begin{equation}\label{eq:deltatheta}
        \Delta\theta_{\max} = \left|\hat{\theta}_{\rm e} - \theta_{\rm b}\right|.
      \end{equation}
\end{itemize}
Clearly, to have smaller SOR, we expect to reduce both $\bar{d}_{{\rm e},0}^{\rm uj}$ and $\Delta \theta_{\max}$. To minimize $\bar{d}_{{\rm e},0}^{\rm uj}$ in \eqref{eq:43.2}, we need to minimize $C_1(\phi)$ while maximizing $C_2(\phi)$; however, since $C_3(\phi) = \frac{C_2(\phi)}{C_1(\phi)}$, this results in a maximized $C_3(\phi)$ (correspondingly, a larger $\Delta\theta_{\max}$), which is not desired. It is clear that a trade-off in $\phi$ exists in balancing the effects of both $\bar{d}_{{\rm e},0}^{\rm uj}$ and $\Delta \theta_{\max}$, which can be formulated as jamming power allocation problems, as described in the following sections.

At last, by setting $\phi = 0$ in Proposition \ref{prop:sor}, we obtain the following corollary:
\begin{corollary}\label{coro:nj sor}
Without jamming, i.e., $\phi = 0$, the SOR can be found as
\begin{equation}\label{eq:sor nj}
  \mathcal{R}_{\rm SOR}^{\rm nj} = \left\{(\theta_{\rm e}, d_{\rm e})\left.\right|d_{\rm e} < \left(\frac{\tilde{P}_{\rm tot}N_t 2^{R_{\rm th}} s_{\rm e;b}\left(\theta_{\rm e}\right)}{{1 + \tilde{P}_{\rm tot}d^{-\alpha}_b N_t} - 2^{R_{\rm th}}}\right)^{\frac{1}{\alpha}}\right\}
\end{equation}
where the superscript $(\cdot)^{\rm nj}$ stands for ``no jamming''.
\end{corollary}

\begin{proof}
  The corollary is directly obtained by setting $\phi = 0$ in Proposition \ref{prop:sor}.
\end{proof}

Differently from \eqref{eq:38}, the constraint on $s_{\rm e;b}(\theta_{\rm e})$ vanishes in \eqref{eq:sor nj}, indicating that the SOR now is extended to the entire angular domain. Moreover, compared with \eqref{eq:43.2}, we see that $\bar{d}_{{\rm e},0}$ in no-jamming case,\footnote{Hereafter, for the conditions where the corresponding notation can be applied for either the uniform-jamming or no-jamming cases, we ignore the superscript ``uj'' or ``nj'' for brevity.} on the contrary, is reduced compared to uniform jamming. In conclusion, uniform jamming induces two opposite effects: the beneficial one is that the SOR can be squeezed in angular domain, and the disadvantage is that the SOR is enlarged in Bob's direction, i.e., the main lobe. Illustration of the SOR changing caused by jamming will be shown later in simulations.

\subsection{SOP Analysis}
With a single Eve uniformly distributed in $\mathcal{R}_{\rm sus}$, using the derived SOR, the SOP is given by
\begin{equation}\label{eq:area Pout}
  \mathcal{P}_{\rm out,singleEve} = \frac{{\rm Area}\left(\mathcal{R}_{\rm SOR}(\phi)  \cap \mathcal{R}_{\rm sus}\right)}{{\rm Area}\left(\mathcal{R}_{\rm sus}\right)}
\end{equation}
where ${\rm Area}(\cdot)$ denotes the area of a certain geometric region. Considering that there are $L$ Eves uniformly distributed in $\mathcal{R}_{\rm sus}$, the SOP of the entire network can be written as
\begin{equation}\label{eq:general Pout}
  \mathcal{P}_{\rm out} = 1 - \left(1 - \mathcal{P}_{\rm out,singleEve}\right)^L.
\end{equation}
 From \eqref{eq:area Pout}, the SOP is determined by the overlapping area between two geometrical regions. If $\mathcal{R}_{\rm SOR}(\phi)  \cap \mathcal{R}_{\rm sus} = \emptyset $, zero SOP is achieved. Recalling \eqref{eq:Rsus}, as well as \eqref{eq:43.2} and \eqref{eq:deltatheta}, two sufficient conditions of $\mathcal{R}_{\rm SOR}(\phi)  \cap \mathcal{R}_{\rm sus} = \emptyset$ can be written as
\begin{align}\label{eq:34.2}
  &{\rm D}_{\min}(\theta_{\rm e}) > \bar{d}_{{\rm e},0},~\forall \theta_{\rm e},\nonumber\\
  &{\rm or}~
  {\rm D}_{\max}(\theta_{\rm e}) = 0, ~\forall \left|\theta_{\rm e} - \theta_{\rm b}\right| < \Delta\theta_{\max}
\end{align}
where $\bar{d}_{{\rm e},0}$ is defined in \eqref{eq:43.2} and $\Delta\theta_{\max}$ is in \eqref{eq:deltatheta}. The physical insight of \eqref{eq:34.2} is clear: when an Eve is far away or  its angle difference to $\theta_{\rm b}$ is large, it does not cause outage.

For the general case with arbitrary shape of $\mathcal{R}_{\rm sus}$, $\mathcal{P}_{\rm out}$ in \eqref{eq:area Pout} can be numerically evaluated and further applied to jamming power allocation design. However, due to the non-regular shapes of $\mathcal{R}_{\rm SOR}(\phi)$ and $\mathcal{R}_{\rm sus}$, closed-form expressions of ${\rm Area}\left(\mathcal{R}_{\rm SOR}(\phi) \cap \mathcal{R}_{\rm sus}\right)$ as well as the SOP in \eqref{eq:general Pout} do not exist for the general case.
Yet, by considering constant boundaries of $\mathcal{R}_{\rm sus}$ as described in Definition 1, \eqref{eq:area Pout} can be written in an integral form as the following proposition.

\begin{prop}\label{prop:SOP}
  With constant boundaries of $\mathcal{R}_{\rm sus}$, i.e., $\mathcal{A}_{\rm e} = [\theta_{\min}, \theta_{\max}]$ and $\mathcal{D}_{\rm e} = [d_{\rm min}, d_{\rm max}]$, and uniform jamming in ${\rm null}(\mathbf{h}_{\rm b})$, the SOP can be given as
  \begin{multline}\label{eq:44}
    \mathcal{P}_{\rm out} = 1 -
    \Bigg\{\int_{d_{\min}}^{d_{\max}} F_{s_{{\rm e};{\rm b}}}\left(\frac{z^{\alpha} + \tilde{P}_{\rm jam}}{\frac{\tilde{P}_{\rm b} N_t 2^{R_{\rm th}}}{{1 + \tilde{P}_{\rm b} d^{-\alpha}_b N_t} - 2^{R_{\rm th}}} + \tilde{P}_{\rm jam}}\right)\\
    \times \frac{2z}{d_{\max}^2 - d_{\min}^2} {\rm d}z\Bigg\}^L.
  \end{multline}
  where $F_{s_{{\rm e};{\rm b}}}(\cdot)$ is defined in \eqref{eq:19}.
\end{prop}

\begin{proof}
For the ease of analytical description, herein we utilize the CDF of the normalized crosstalk in \eqref{eq:19}.
First, rewrite \eqref{eq:SOP} as ${\mathcal P}_{\rm out} = 1 - \Pr\{R^{\rm s}_{\rm b} \geq R_{\rm th}\}$
and recall \eqref{eq:sec rate}, we have
  \begin{equation}\label{eq:46.2}
    \Pr\{R^{\rm s}_{\rm b} \geq R_{\rm th}\}
    = F_{{\rm SINR}_{{\rm e},{\max}}^{\rm uj}}\left(\frac{1 + {\rm SINR}_{\rm b}^{\rm uj}}{2^{R_{\rm th}}} - 1\right)
  \end{equation}
where $F_{{\rm SINR}_{{\rm e},{\max}}^{\rm uj}}(\cdot)$ is the CDF of ${\rm SINR}_{{\rm e},{\max}}^{\rm uj}$, which is given by $F_{{\rm SINR}_{{\rm e},{\max}}^{\rm uj}}(x) = \left(F_{{\rm SINR}_{{\rm e}}^{\rm uj}}(x)\right)^L$ since all Eves are independently distributed. Using \eqref{eq:30}, we have
\begin{equation}\label{eq:27.1}
  F_{{\rm SINR}_{{\rm e}}^{\rm uj}}(x)
  = \Pr\left\{s_{{\rm e};{\rm b}}(\theta_{\rm e}) \leq \frac{{d^{\alpha}_{{\rm e}}} + \tilde{P}_{\rm jam}}{\frac{\tilde{P}_{\rm b} N_t}{x} + \tilde{P}_{\rm jam}}\right\}
\end{equation}
where both $\theta_{\rm e}$ and $d_{{\rm e}}$ are random. Since $\theta_{{\rm e}}$ and $d_{{\rm e}}$ are independent, \eqref{eq:27.1} can be presented as
\begin{equation}
  F_{{\rm SINR}_{{\rm e}}}(x)
  = \int_{d_{\min}}^{d_{\max}} F_{s_{{\rm e};{\rm b}}}\left(\frac{z^{\alpha} + \tilde{P}_{\rm jam}}{\frac{\tilde{P}_{\rm b} N_t}{x} + \tilde{P}_{\rm jam}}\right) f_{d_{\rm e}}(z) {\rm d}z.
\end{equation}
where $f_{d_{\rm e}}(z) \triangleq \frac{2z}{d_{\max}^2 - d_{\min}^2}$ is the PDF of $d_{{\rm e}}$, corresponding to the uniform distribution between two boundaries defined by $\mathcal{D}_{\rm e} = [d_{\rm min}, d_{\rm max}]$. Then,
$\mathcal{P}_{\rm out}$ is directly obtained as \eqref{eq:44}.
\end{proof}

Practically, \eqref{eq:44} can be used for jamming power allocation. As stated in Remark 1, Alice can arbitrarily adjust the value of the constant boundaries in the design, based on the information about $\mathcal{R}_{\rm sus}$ that she has. Particularly, if Alice knows nothing about $\mathcal{R}_{\rm sus}$ (i.e., she assumes $\mathcal{A}_{\rm e} = [0, 2\pi]$ and $\mathcal{D}_{\rm e} = [0, r_{\rm max}]$), minimizing $\mathcal{P}_{\rm out}$ becomes equivalent to minimizing ${\rm Area}\left(\mathcal{R}_{\rm SOR}(\phi)\right)$.

\subsection{Jamming-beneficial Range}
Based on Proposition \ref{prop:SOP}, we find a jamming-beneficial range defined in $d_{\max}$ (i.e., the larger constant distance boundary of $\mathcal{R}_{\rm sus}$) as follows.
\begin{prop}\label{prop:benefit range}
A constraint on $d_{\max}$ that makes the uniform jamming beneficial in reducing the SOP is given by
  \begin{equation}\label{eq:27.3}
    d_{\max} < \left(s_{\rm e;b}^{\rm max} \frac{\tilde{P}_{\rm tot} N_t 2^{R_{\rm th}}}{{1 + \tilde{P}_{\rm tot} d^{-\alpha}_{\rm b} N_t} - 2^{R_{\rm th}}}\right)^{\frac{1}{\alpha}}
  \end{equation}
where $s_{\rm e;b}^{\rm max}$ is the largest feasible crosstalk value defined in Lemma \ref{le:CT cha}.
\end{prop}

\begin{proof}
See Appendix \ref{proof:benificial}.
\end{proof}

\begin{remark}
Proposition \ref{prop:benefit range} shows that when Eves are located close enough to Alice, uniform jamming is always beneficial in reducing the SOP. Clearly, this range expands with larger $d_{\rm b}$, as well as larger $s_{\rm e;b}^{\rm max}$ or larger $R_{\rm th}$. On the other hand, the range shrinks with larger $N_t$ or $\tilde{P}_{\rm tot}$.

Moreover, we note that \eqref{eq:27.3} has a similar form of that described for the SOR without jamming, i.e., $\mathcal{R}_{\rm SOR}^{\rm nj}$ in \eqref{eq:sor nj}. Recalling the definitions of the largest distance of SOR in \eqref{eq:43.2} and \eqref{eq:radius side lobe}, the physical insight of Proposition \ref{prop:benefit range} can be explained as follows: as long as $\mathcal{R}_{\rm sus} \cap \mathcal{R}_{\rm SOR}^{\rm nj} \neq \emptyset$, there always exists an optimal $\phi$, with which the SOP can be reduced by uniform jamming, compared with the SOP without jamming. The optimization of $\phi$ is discussed in the next section.
\end{remark}

\section{Jamming Power Allocation}

In this section, considering uniform jamming, we investigate the optimal jamming power allocation that minimizes the SOP. The problem can be simply described as
\begin{equation}\label{eq:47}
\mathop {\min }\limits_{\phi} \mathcal{P}_{\rm out}, ~~
{\rm{s}}{\rm{.t}}{\rm{.   }}~~0 \leq \phi \leq 1.
\end{equation}
In practice, Alice may have different accuracy levels of information about $\mathcal{R}_{\rm sus}$ as follows:
\begin{enumerate}
  \item Alice knows nothing about the suspicious area, or only partial information about the suspicious area such as $\mathcal{A}_{\rm e}$ only (or $\mathcal{D}_{\rm e}$ only); and
  \item Alice knows exact information about the suspicious area, i.e., both $\mathcal{A}_{\rm e}$ and $\mathcal{D}_{\rm e}$.
\end{enumerate}
For these two cases, we respectively investigate the jamming power allocation in the following.

\subsection{Jamming with None/Partial Information about $\mathcal{R}_{\rm sus}$}
When Alice knows nothing about $\mathcal{R}_{\rm sus}$, minimizing the SOP becomes equivalent to minimizing the area of SOR, which can be calculated as
\begin{equation}\label{eq:area sor 1}
  {\rm Area}\left(\mathcal{R}_{\rm SOR}^{\rm uj}(\phi)\right) = \int_{0}^{2\pi}\frac{1}{2}\left(\bar{d}_{\rm e}^{\rm uj}(\phi, \theta_{\rm e})\right)^2 {\rm d}\theta_{\rm e}
\end{equation}
where $\bar{d}_{\rm e}^{\rm uj}(\phi, \theta_{\rm e})$ is defined in \eqref{eq:39.1}.
Note that $\mathcal{R}_{\rm SOR}^{\rm uj}(\phi)$ is composed of many side lobes. We use $\mathcal{R}_{{\rm SOR},m}^{\rm uj}(\phi)$ to denote the $m$-th side lobe, and $\mathcal{R}_{{\rm SOR},\mathcal{I}}^{\rm uj}(\phi)$ to denote a group of side lobes with indices described by the set $\mathcal{I}$.
For the case that Alice knows $\mathcal{A}_{\rm e}$ or $\mathcal{D}_{\rm e}$, we can simplify the problem by minimizing partial, other than the entire area of $\mathcal{R}_{\rm SOR}^{\rm uj}(\phi)$ such as
\begin{equation}\label{eq:partial area}
  {\rm Area}\left(\mathcal{R}_{{\rm SOR},\mathcal{I}'}^{\rm uj}(\phi)\right) = \sum\limits_{m \in \mathcal{I}'} {\rm Area}\left(\mathcal{R}_{{\rm SOR},m}^{\rm uj}(\phi)\right)
\end{equation}
where $\mathcal{I}'$ is the set of the concerned side lobe indices, determined by either $\mathcal{A}_{\rm e}$ or $\mathcal{D}_{\rm e}$. Using \eqref{eq:area sor 1} (or \eqref{eq:partial area}) along with \eqref{eq:39.1}, the areas can be numerically calculated and the optimal $\phi$ can be easily founded via one dimensional linear search.
Since it is difficult to derive closed-form expression for ${\rm Area}\left(\mathcal{R}_{{\rm SOR},\mathcal{I}}^{\rm uj}(\phi)\right)$, we evaluate the area of $\mathcal{R}_{{\rm SOR},m}^{\rm uj}(\phi)$ in the following corollary for a special case to further provide some discussions.

\begin{corollary}[Area of $\mathcal{R}_{{\rm SOR},m}^{\rm uj}(\phi)$]\label{coro:area side lobe}
With $\theta_{\rm b} = 0$ and the path loss coefficient being $\alpha = 2$ (which corresponds the free space propagation \cite{3GPP}), ${\rm Area}\left(\mathcal{R}^{\rm uj}_{{\rm SOR},m}\right)$ can be upper bounded as
\begin{equation}\label{eq:area side lobe}
  {\rm Area}\left(\mathcal{R}^{\rm uj}_{{\rm SOR},m}(\phi)\right) \leq \frac{1}{4N_t d}\left(\frac{K_{\rm e;b}}{\pi^2 m^2}C_1(\phi) - 2C_2(\phi)\right)
\end{equation}
where $C_1(\phi)$ and $C_2(\phi)$ are defined in \eqref{eq:40.1} and \eqref{eq:40.2}.
\end{corollary}

\begin{proof}
  See Appendix \ref{proof:coro2}.
\end{proof}

\begin{remark}\label{rk:4}
  In \eqref{eq:area side lobe}, it is shown that the area of every side lobe is inversely proportional to $N_t$, indicating that the SOR side lobes can asymptotically vanish with ultimately large $N_t$. Moreover, it is inversely proportional to $m^2$, which means that the area of the SOR will rapidly decrease for the side lobes with large indices, i.e., with large angle difference to $\theta_{\rm b}$. This result indicates that Eves from different directions (i.e., within different side lobes) have different significance in causing secrecy outage, hence should be treated differently in the jamming design.

\end{remark}

\subsection{Jamming with Exact Information of $\mathcal{R}_{\rm sus}$}
With the information of $\mathcal{R}_{\rm sus}$, Alice can calculate and apply the value of $\mathcal{P}_{\rm out}$ in the design (at least numerically),\footnote{The calculation requires the knowledge of $\mathcal{D}_e$ and $\mathcal{A}_e$. Clearly, uniform jamming is not optimal in this condition. However, for the ease of analysis, we first devise the optimal power allocation for uniform jamming; then, the resulted jamming power can be allocated directionally to further improve efficiency.} using either \eqref{eq:general Pout} or \eqref{eq:44}.
Although in practice, \eqref{eq:47} can be readily solved by one dimensional linear search, it fails to provide the optimal $\phi$ in closed form.
In the following corollary, we provide closed-form solutions and discussions for a special case.

\begin{corollary}[Jamming power allocation for given $\theta_{\rm e}$]\label{coro:equal theta}
  For constant boundaries of $\mathcal{R}_{\rm sus}$ and given $\theta_{\rm e}$, which is equal for all Eves, the optimal $\phi$ can be determined as
  \begin{equation}\label{eq:phi opt}
    \phi^{\tt opt} = \left\{
    {\begin{array}{*{20}{c}}
{\phi_g^{\tt opt},}&{ \phi_0 \notin [0,1]}\\
{\min\{\phi_g^{\tt opt}, \phi_0\}}&{\phi_0 \in [0,1]}
\end{array}}
\right.
  \end{equation}
  where
  \begin{align}\label{eq:35}
  \phi^{\tt opt}_g &= 1 - \frac{(2^{R_{\rm th}} - 1)+ \sqrt{\frac{(2^{R_{\rm th}} - 1)2^{R_{\rm th}}s_{\rm e;b}(\theta_{\rm e})}{1-s_{\rm e;b}(\theta_{\rm e})}N_t}}{ d^{-\alpha}_{\rm b} N_t \tilde{P}_{\rm tot}},\\
  \label{eq:35.1}
  \phi_0 &= \frac{ \frac{{{{\tilde P}_{{\rm{tot}}}}{N_t}{2^{{R_{{\rm{th}}}}}}}s_{\rm e;b}(\theta_{\rm e})}{{{\rm{1 + }}{{\tilde P}_{{\rm{tot}}}}d_{\rm{b}}^{ - \alpha }{N_t} - {2^{{R_{{\rm{th}}}}}}}} - d_{\min}^{\alpha}}{(1 - s_{\rm e;b}(\theta_{\rm e})){\tilde P}_{\rm tot}}.
\end{align}
\end{corollary}

\begin{proof}
  See Appendix \ref{proof:coroET}.
\end{proof}

Note when $\phi^{\tt opt} = \phi^{\tt opt}_g$, the optimal jamming power decreases with $d_{\rm b}$ and $s_{\rm e;b}(\theta_{\rm e})$, whereas it will increase when $\phi^{\tt opt} = \phi_0$.
The part that dominates the final result in \eqref{eq:phi opt} depends on the value of $\theta_{\rm e}$.
Detailed discussions will be provided in Section VII along with simulations.

\section{Directional Jamming}

In this section, we propose directional jamming algorithms to allocate jamming power more efficiently than uniform jamming, based on the following facts:
\begin{enumerate}
  \item
  With the information of $\mathcal{R}_{\rm sus}$, Alice can perform jamming only to the suspicious directions instead of the entire null space of $\mathbf{h}_{\rm b}$.
  \item
  Without information of $\mathcal{R}_{\rm sus}$, jamming towards different directions also needs to be treated differently, as stated in Remark \ref{rk:4}.
\end{enumerate}
At a cost of slightly increasing the implementation complexity compared with uniform jamming, directional jamming is able to substantially reduce the SOP. In following subsections, we present power allocation algorithms for directional jamming with and without the information of $\mathcal{R}_{\rm sus}$.

\subsection{Directional Jamming with the Information of $\mathcal{R}_{\rm sus}$}

When jamming is not uniformly performed, from \eqref{eq:7}, the SINR at Eve is represented as
\begin{equation}\label{eq:SINR nj}
  {\rm SINR}_{\rm e}^{\rm dj} = \frac{P_{\rm b}d^{-\alpha}_{\rm e}|\mathbf{h}_{\rm e}^H \mathbf{w}_{\rm b}|^2}{N_0 + d^{-\alpha}_{\rm e} \mathbf{h}_{\rm e}^H \mathbf{V}{\rm diag}\left(\bar{\mathbf{p}}\right)\mathbf{V}^H \mathbf{h}_{\rm e}}
\end{equation}
where $\mathbf{V}\in \mathbb{C}^{N_t \times N}$ is the matrix that spans the jamming space, with the $n$-th column vector being $\mathbf{v}_n$, and $\bar{\mathbf{p}} = \left(\bar{P}_1,...,\bar{P}_N\right)^T$, where $\bar{P}_n$ is the power allocated to the $n$-th jamming direction as defined in \eqref{eq:leg sig mod}. Correspondingly, the SOR now can be described as
\begin{align}\label{eq:SOR dj}
  \mathcal{R}_{\rm SOR}^{\rm dj}(\bar{\mathbf{p}})
  &= \left\{(d_{\rm e}, \theta_{\rm e})\left.\right|d_{\rm e} \leq \bar{d}_{\rm e}^{\rm dj}(\bar{\mathbf{p}},\theta_{\rm e}) \right\}\\
\label{eq:70}
 \bar{d}_{\rm e}^{\rm dj}(\bar{\mathbf{p}},\theta_{\rm e}) &= \Bigg[\bigg(\frac{2^{R_{\rm th}}(1 - \phi)\tilde{P}_{\rm tot}N_ts_{\rm e;b}(\theta_{\rm e})}{{1 + (1 - \phi)\tilde{P}_{\rm tot}d^{-\alpha}_b N_t} - 2^{R_{\rm th}}} \nonumber\\
 &~~~~- \bar{\mathbf{s}}(\theta_{\rm e})^H \mathbf{V}{\rm diag}(\tilde{\mathbf{p}})\mathbf{V}^H \bar{\mathbf{s}}(\theta_{\rm e})\bigg)^{\frac{1}{\alpha}}\Bigg]^+
\end{align}
where $\tilde{\mathbf{p}} \triangleq \frac{\bar{\mathbf{p}}}{N_0}$, $\bar{\mathbf{s}}(\cdot)$
was defined in \eqref{eq:2}.
The superscript $(\cdot)^{\rm dj}$ stands for ``directional jamming''.

Design directional jamming using \eqref{eq:SINR nj} induces high complexity especially when $N$ is large, since changing any element in the $N$-dimensional vector $\bar{\mathbf{p}}$ requires re-calculation of $\mathcal{R}_{\rm SOR}^{\rm dj}(\bar{\mathbf{p}})$.
Hence, we alternatively propose a two-step suboptimal power allocation method for directional jamming in Algorithm \ref{alg:1}, which firstly find the optimal jamming power assuming uniform jamming, then reallocate it directionally based on a criterion of jamming subspace selection.
\begin{algorithm}[tbp]
\caption{Directional jamming with the information of $\mathcal{R}_{\rm sus}$}
\begin{algorithmic}[1]
\State {\bf Initialization:} Update the information of $\theta_{\rm b}$, $d_{\rm b}$ and $\mathcal{R}_{\rm sus}$ at Alice.
\State Assuming null space-based uniform jamming, find
  \begin{equation}\nonumber
    \phi^{\tt opt} = \arg\mathop {\min }\limits_{\phi} \mathcal{P}_{\rm out}
  \end{equation}
  through one dimensional linear search over $\phi \in [0,1)$.
\State Select a $N$-dimensional subspace ($\mathbf{v}_n, n = 1,...,N$ and $N \leq N_t$) from ${\rm null}\{\mathbf{h}_{\rm b}\}$ according to \eqref{eq:43}, where $\theta_{\mathbf{v}_n}$ is defined as \eqref{eq:42}.
\State Equally allocate $P_{\rm jam}^{\tt opt} = \phi^{\tt opt}P_{\rm tot}$ to the selected beams in step $3$ such as $\frac{P_{\rm jam}^{\tt opt}}{N}$.
\end{algorithmic}\label{alg:1}
\end{algorithm}
In Step $2$, $\mathcal{P}_{\rm out}$ can be calculated numerically using \eqref{eq:general Pout} or \eqref{eq:44}, depending on the available information of $\mathcal{R}_{\rm sus}$. Note that $\theta_{\rm b}$, $d_{\rm b}$ and $\mathcal{R}_{\rm sus}$ are long term parameters, thus the updating period of Step $1$ can be much longer than the computation time required by the other steps.

 Since $\mathcal{R}_{\rm sus}$ is defined by physical angles, it is necessary to set up a mapping between the jamming space and physical angle to concentrate the jamming power towards $\mathcal{R}_{\rm sus}$. We propose a heuristic subspace selection method for Step $3$.
First, map $\mathbf{v}_n$ to physical angle as
\begin{equation}\label{eq:42}
  \theta_{\mathbf{v}_n} = \arg  \mathop {\max}\limits_{\theta \in [-\frac{\pi}{2},\frac{\pi}{2}]} \left|\bar{\mathbf{s}}(\theta)^T \mathbf{v}_n\right|^2.
\end{equation}
After that, $P_{\rm jam}^{\tt opt}$ is equally reallocated to the beams whose indices are
\begin{equation}\label{eq:43}
  \mathcal{N} \triangleq \{n \left.\right| \theta_{\mathbf{v}_n} \in \mathcal{A}_{\rm e}\}.
\end{equation}
The power allocated to each beam is now $\frac{P_{\rm jam}^{\tt opt}}{\dim(\mathcal{N})}$.
In practice, $\mathbf{v}_n$ is not necessarily in ${\rm null}(\mathbf{h}_{\rm b})$ and an alternative is to find $\mathbf{v}_n$ as the column vectors of a $N_t$-dimensional DFT matrix for the following reasons; 1) selected columns of the DFT matrix can form a good substitute of ${\rm null}(\mathbf{h}_{\rm b})$, as $N_t \to \infty$ \cite{JSDM}; 2) using pre-defined DFT basis as the jamming space avoids channel inverse calculation, which induces high computation complexity especially when $N_t$ is large \cite{Zhu:Arxiv14}; and
3) most importantly, the structure of DFT matrix provides very sharp beam pattern towards the physical angle $\theta_{\mathbf{v}_n}$ in \eqref{eq:42}, therefore, the beam selection criterion \eqref{eq:43} can be very efficient since with sharper beams, there will be less jamming power leaked outside of $\mathcal{R}_{\rm sus}$.

\subsection{Directional Jamming without Information of $\mathcal{R}_{\rm sus}$}
Without any information of $\mathcal{R}_{\rm sus}$,
the objective of directional jamming power allocation becomes to minimize the area of $\mathcal{R}_{\rm SOR}^{\rm dj}(\bar{\mathbf{p}})$ in \eqref{eq:SOR dj} for $\theta_{\rm e} \in [0,2\pi]$, which is calculated as
\begin{equation}\label{eq:71}
  {\rm Area}\left(\mathcal{R}_{\rm SOR}^{\rm dj}(\bar{\mathbf{p}})\right) = \int_0^{2\pi} \frac{1}{2}\left(\bar{d}_{\rm e}^{\rm dj}(\bar{\mathbf{p}}, \theta_{\rm e})\right)^2 {\rm d}\theta_{\rm e}.
\end{equation}
A general closed-form expression of \eqref{eq:71} is not available, and its convexity is unknown. Hence, numerically minimizing ${\rm Area}\left(\mathcal{R}_{\rm SOR}^{\rm dj}(\bar{\mathbf{p}})\right)$ is NP-hard. To overcome this, we propose Algorithm 2, which iteratively finds the optimal $n$-th element of $\bar{\mathbf{p}}$ while keeping the others fixed.

\begin{algorithm}[tbp]
  \caption{Iterative directional jamming without information of $\mathcal{R}_{\rm sus}$}
  \begin{algorithmic}[1]
    \State {\bf Initialization:} Update the information of $\theta_{\rm b}$ and $d_{\rm b}$ at Alice; set an initial jamming power allocation vector $\bar{\mathbf{p}}_0$, satisfying $||\bar{\mathbf{p}}_0|| \leq P_{\rm tot}\phi^{\rm max}$; set the initial iteration index $j = 1$.
    \For {$n = 1$ to $N$}
    \State Update
      \begin{multline}\nonumber
      \bar{\mathbf{p}}_j(n) = \mathop {\arg \min }\limits_{x\in [0,x_{\rm max} )} {\rm Area}\Bigg(\mathcal{R}_{\rm SOR}^{\rm dj}\bigg([\bar{\mathbf{p}}_j(1),...,\\
      \bar{\mathbf{p}}_j(n-1),x,\bar{\mathbf{p}}_{j-1}(n+1),...,
      \bar{\mathbf{p}}_{j-1}(N)]\bigg)\Bigg)
      \end{multline}
      ~~~~where $x_{\rm max} = P_{\rm tot}\phi^{\rm max}-\sum\limits_{i=1}^{n-1}\bar{p}_j(i) - \sum\limits_{i = n+1}^N \bar{p}_{j - 1}(i)$.
    \EndFor
    \If {$||\bar{\mathbf{p}}_j - \bar{\mathbf{p}}_{j-1}|| \geq \varepsilon$}
      \State $j = j + 1$; {\bf go to} step 2.
    \ElsIf {$||\bar{\mathbf{p}}_j - \bar{\mathbf{p}}_{j-1}|| < \varepsilon$}
      \State \Return
    \EndIf
  \end{algorithmic}\label{alg:2}
\end{algorithm}

Algorithm \ref{alg:2} provides a sub-optimal solution which reduces the complexity by degrading the original problem to one-dimensional linear search. However, for large $N_t$, the complexity is still huge since during each main iteration, $N \sim \mathcal{O}(N_t)$ times of linear searching are required to fully update $\bar{\mathbf{p}}$. Hence, Algorithm \ref{alg:2} is not suitable for some scenarios where $\theta_{\rm b}$ and $d_{\rm b}$ change fast. In this light, we propose a simplified algorithm, Algorithm 3, to further reduce the complexity. In Step 2 of Algorithm 3, $\theta_m$ is the mean angle of the $m$-th side lobe of $\mathcal{R}_{\rm SOR}^{\rm dj}(\bar{\mathbf{p}})$ ($m = 0$ denotes the main lobe). In Step 4, $\bar{\mathbf{p}}_{\rm boundary}$ follows the structure such as
\begin{align}\label{eq:boundary}
  \bar{{\mathbf{p}}}_{\rm boundary} &= \left\{0, ..., \bar{P}_{m'_1}, 0,...,\bar{P}_{m'_2},0,...\right\}, \nonumber\\
  \bar{P}_{m'_1} + \bar{P}_{m'_2} &= \phi{P}_{\rm tot}
\end{align}
where $m'_1$ and $m'_2$ are the indices of the two {\em dominating side lobes}, which are located most closely to the main lobe (from both sides).
 The derivation of Algorithm 3 is described in Appendix \ref{proof:alg3}.

With Algorithm $3$, jamming is performed in only two dominating directions in the neighborhood of $\theta_{\rm b}$, for the reasons that 1) for the region with large angle difference to Bob, allocating much jamming power is inefficient since $\mathcal{R}_{\rm SOR}^{\rm nj}$ in this region is generally very small; 2) for the directions highly in-line with $\theta_{\rm b}$, jamming should be avoided as it will cause severe interference to Bob.
Note that for every realization of $\phi$, only single time of linear search is required in Step $4$. The complexity is irrelative to $N$ (which is large in general), hence can be greatly reduced compared to Algorithm \ref{alg:2}. As a possible extension, more than two dominating directions can be involved in the design while the trade-off between complexity and performance exists.

\begin{algorithm}[tbp]
  \caption{Simplified directional jamming without information of $\mathcal{R}_{\rm sus}$}
  \begin{algorithmic}[1]
    \State {\bf Initialization:} Update the information of $\theta_{\rm b}$ and $d_{\rm b}$ at Alice; Initialize ${\rm Area}_{\min} = {\rm Inf}.$;
    \State Calculate $s_{{\rm e;b}}(\theta_{m}), m = 1,...,M$. Determine $m'_1, m'_2$, satisfying that $s_{{\rm e;b}}(\theta_{m'_1}) \geq s_{{\rm e;b}}(\theta_{m'_2}) \geq s_{{\rm e;b}}(\theta_{m}), \forall m \neq m'_1$ and $m\neq m'_2$;
    \For {$\phi = 0$ to $1$}
    \State
    Find
    \begin{equation}\nonumber
      \bar{\mathbf{p}}_{\rm boundary}^{\tt opt} = {\rm arg}\min {\rm Area}\left(\mathcal{R}_{\rm SOR}^{\rm dj}\left(\bar{\mathbf{p}}_{\rm boundary}\right)\right).
    \end{equation}
    \State
    Calculate ${\rm Area_{\rm SOR}^{\rm dj}\left(\bar{\mathbf{p}}_{\rm boundary}^{\tt opt}\right)}$ according to \eqref{eq:71};
    \If {${\rm Area_{\rm SOR}^{\rm dj}\left(\bar{\mathbf{p}}_{\rm boundary}^{\tt opt}\right)} < {\rm Area}_{\min}$}
      \State ${\rm Area}_{\min} = {\rm Area_{\rm SOR}^{\rm dj}\left(\bar{\mathbf{p}}_{\rm boundary}^{\tt opt}\right)}$;
      \State $\phi^{\tt opt} = \phi$; $\bar{p}^{\tt opt} = \bar{\mathbf{p}}_{\rm boundary}^{\tt opt}$.
    \Else
      \State {\bf continue}
    \EndIf
    \EndFor
  \end{algorithmic}\label{alg:3}
\end{algorithm}

{
\section{Extension to Multiuser and Multi-cell Scenarios}
We focus on the single-cell and single-user scenario in previous sections. In this section, we now show how the SOR can be affected by multiple users and cells. We also provide discussions on the design of secure transmission in these scenarios with future research challenges.

\subsection{Multiuser Transmission}
When multiple legitimate users (i.e., multiple Bobs) are presented in massive MIMO systems for Rician channels,
the multiuser interference between Bobs is trivial as long as their LOS angles have large difference, which can be readily ensured via user scheduling. On the contrary, the multiuser interference to Eves can be seen as equivalent jamming considering that single-user decoder is adopted at Eves, which is likely to happen when Eves are low-cost devices. Hence, when multiuser beamforming is applied for Bobs, the received multiuser interference at Eve becomes equal to the directional jamming, transmitted towards other Bobs' directions. Consequently, the SOR of an objective Bob will be shrunk in the directions of the other Bobs.
Denote the set of all legitimate users as $\mathcal{I}_{\rm r}^{\rm MU} = \{{\rm b}_1, ..., {\rm b}_{U}\}$. Similar to \eqref{eq:SOR dj} and \eqref{eq:70}, which describe the SOR for directional jamming, the SOR of user ${\rm b}_u \in \mathcal{I}_{\rm r}^{\rm MU}$ in the presence of multiple users can now be described as
\begin{align}
  \mathcal{R}_{{\rm SOR}}^{{\rm MU},{\rm b}_u}(\bar{\mathbf{p}}_{{\rm b}_u})
  = \left\{(d_{\rm e}, \theta_{\rm e})\left.\right|d_{\rm e} \leq \bar{d}_{\rm e}^{{\rm MU},{\rm b}_u}(\bar{\mathbf{p}}_{{\rm b}_u},\theta_{\rm e}) \right\}
\end{align}
\begin{multline}
 \bar{d}_{\rm e}^{{\rm MU},{\rm b}_u}(\bar{\mathbf{p}}_{{\rm b}_u},\theta_{\rm e}) = \Bigg[\bigg(\frac{2^{R_{\rm th}}(1 - \phi)\tilde{P}_{\rm tot}N_ts_{{\rm e;}{\rm b}_u}(\theta_{\rm e})}{{1 + (1 - \phi)\tilde{P}_{\rm tot}d^{-\alpha}_{{\rm b}_u} N_t} - 2^{R_{\rm th}}}\\
  - \bar{\mathbf{s}}(\theta_{\rm e})^H \mathbf{V}{\rm diag}(\tilde{\mathbf{p}}_{{\rm b}_u})\mathbf{V}^H \bar{\mathbf{s}}(\theta_{\rm e})\bigg)^{\frac{1}{\alpha}}\Bigg]^+
\end{multline}
where $\mathbf{V} = \left[\mathbf{W}_{\rm b},\mathbf{V}_{\rm j}\right] \in \mathbb{C}^{N_t \times N}$ spans the equivalent jamming space, in which $\mathbf{W}_{\rm b} = \left[\mathbf{w}_{{\rm b}_{u'}}\right],u' = 1,...,U, u'\neq u$ spans the signaling space for the other legitimate users, while $\mathbf{V}_{\rm j} \in {\rm null}\left(\mathbf{W}_{\rm b}\right)$ is the jamming space. Correspondingly, the power allocation vector can be divided into two parts as
  $\tilde{\mathbf{p}}_{{\rm b}_u} = \left[\tilde{\mathbf{p}}_{{\rm b}_u,{\rm sig}}^T, \tilde{\mathbf{p}}_{{\rm b}_u,{\rm jam}}^T\right]^T$
where
$\tilde{\mathbf{p}}_{{\rm b}_u,{\rm sig}} \triangleq \left[\tilde{P}_{{\rm b}_1},...,\tilde{P}_{{\rm b}_{u-1}}, \tilde{P}_{{\rm b}_{u+1}}, ..., \tilde{P}_{{\rm b}_U}\right]^T$
is the signal power allocation vector for all legitimate users except for ${\rm b}_u$. For given fixed $\tilde{\mathbf{p}}_{{\rm b}_u,{\rm sig}}$,
in order to minimize ${\rm Area}\left(\mathcal{R}_{{\rm SOR}}^{{\rm MU},{\rm b}_u}(\bar{\mathbf{p}}_{{\rm b}_u})\right)$ for user ${\rm b}_u$, the directional jamming algorithms, i.e., Algorithm $2$ and $3$, can be directly applied herein.

Considering communication secrecy for the entire multiuser transmission system, the optimization problem can be reasonably re-formulated as a min-max problem such as
\begin{equation}\label{eq:55.1}
  \begin{array}{l}
\mathop {\min }\limits_{{{{\bf{\bar p}}}_{{{\rm{b}}_u}}}} \mathop {\max }\limits_{{{\rm{b}}_u}} {\rm{Area}}\left( {\mathcal{R}_{{\rm{SOR}}}^{{\rm{MU}},{{\rm{b}}_u}}\left( {{{{\bf{\bar p}}}_{{{\rm{b}}_u}}}} \right)} \right)\\
{\rm{s}}{\rm{. t}}{\rm{. }}~~{\left\| {{{{\bf{\bar p}}}_{{{\rm{b}}_u}}}} \right\|_1} < {P_{{\rm{tot}}}},{{\rm{b}}_u} \in \mathcal{I}_{\rm r}^{\rm MU}.
\end{array}
\end{equation}
The main challenge in solving \eqref{eq:55.1} is that allocating power for one user affects the SORs of other users.
Hence, the power needs to be jointly allocated for all users, and the complexity of such joint optimization can be very high. Hence, it is desirable to develop simplified algorithms for the multiuser scenario.

\begin{figure}[tbp]
  \centering
  \psfrag{x}[][][0.7]{Jamming power allocation coefficient $\phi$}
  \psfrag{y}[][][0.7]{Radius of the lobes of $\mathcal{R}_{\rm SOR}^{\rm uj}$ (m)}
  \psfrag{main lobe}[l][l][0.7]{main lobe}
  \psfrag{the m th side lobes m 1 2 3}[l][l][0.7]{the $m$-th side lobes, $m = 1, 2, 3...$}
  \includegraphics[width=1\columnwidth,height = 2.5in]{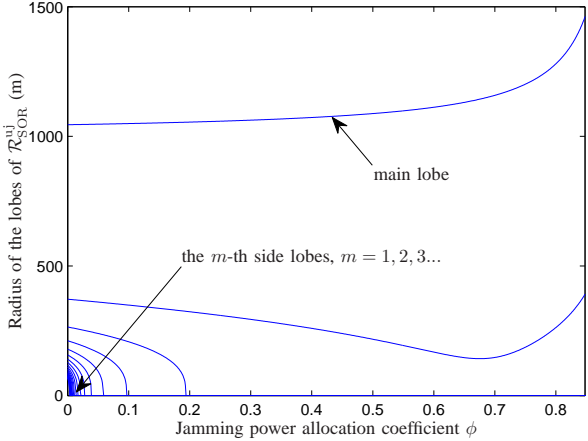}
  \caption{Radius of the main lobe/side lobes of $\mathcal{R}_{\rm SOR}^{\rm uj}$ vs. $\phi$. Parameters setting: $\theta_{\rm b} = 0^\circ$, $d_{\rm b} = 100$m, $N_t = 100$, and $R_{\rm th} = 10$bps/Hz.}
  \label{fig:radius}
\end{figure}

\subsection{Multi-cell Network}
In multi-cell massive MIMO networks, it is commonly assumed that the training pilots are reused among cells. Correspondingly, pilot contamination results in imperfect CSI estimation as well as nonnegligible multi-cell interference from Alices in adjacent cells. Denote $\theta_{i,{\rm b}}^j$ as the angular direction of Bob in Cell $i$ seen from Alice in Cell $j$, and Cell $0$ as the objective cell where Bob $0$ exists. Some major effects of imperfect CSI and multi-cell interference on the SOR of Bob $0$ are described as follows:
\begin{itemize}
  \item Due to pilot contamination from Bobs $i$, $\forall i \neq 0$, the SOR of Bob $0$ will be enlarged in the directions of $\theta_{i,{\rm b}}^0$.
  \item Multi-cell interference to Bob $0$ will isotropically enlarge his SOR. On the other hand, Eves in Cell $0$ are equivalently jammed by the multi-cell interference from Alice $j$ in cell $j$, $\forall j \neq 0$, especially in the directions of $\theta^j_{0,{\rm b}}$ and $\theta^j_{j,{\rm b}}$. Thereby, the SOR in these directions can be shrunk and the shape can be non-continuous in these regions.
\end{itemize}

A general analytical description of the SOR in the multi-cell network is challenging, since it is determined not only by the network topology, but also by the locations of all pilot-contaminating users in adjacent cells. Moreover, the complicated shape of the SOR makes difficulties in calculating and minimizing the corresponding area. Nevertheless,
in practice, pilot scheduling and reuse schemes can be utilized to alleviate these adverse effects which are caused by pilot contamination, e.g., \cite{Yin,You}. }

\section{Simulation Results}

Simulation results are shown in this section.
We set $d_0 = 0.5$, $\alpha = 3$, $P_{\rm tot} = 1{\rm W}$, $N_0 = 10^{-5}$${\rm mW}$ and for simplicity, assume strong LOS environment such that $K_{\rm e;b}\to 1$. In this parameter setting, the receive SNR is 20dB when the transmitter-receiver distance is 100m.

At first, using \eqref{eq:43.2} and \eqref{eq:radius side lobe}, Fig.~\ref{fig:radius} provides a description of $\mathcal{R}_{\rm SOR}^{\rm uj}(\phi)$ in terms of the radius of the main lobe/side lobes. It is shown that the radius of the main lobe is monotonously increasing with $\phi$, indicating that reducing the signal power towards $\theta_{\rm b}$ enlarges the SOP in this direction. Clearly, allocating additional jamming power to the direction of $\theta_{\rm b}$ will further enlarge this radius, which suggests that jamming directly in the legitimate user's direction should be avoided. This conclusion coincides with the concept we followed for the design of Algorithm 3.

Moreover, as $\phi$ increases, the radius of the second side lobe is reduced first and then increases after a certain value, e.g., $\phi \approx 0.7$, indicating an optimal jamming power allocation in terms of minimizing the SOP {\em in this direction}. The side lobes with index $m \geq 3$ can be completely eliminated with proper jamming.
The results show that we can design jamming based on the partial information of $\mathcal{R}_{\rm sus}$. For example, in Fig.~\ref{fig:radius}, if we know that Eves are located in the direction ranges of side lobes with indices larger than $3$, then, allocating $\phi = 0.2$ is enough to secure the communication and the remaining power can be allocated to data transmission.

\begin{figure}[tbp]
  \centering
  \psfrag{Jamming power allocation coefficient}[][][0.7]{Jamming power allocation coefficient $\phi$}
  \psfrag{Secrecy outage probability}[][][0.7]{Secrecy outage probability}
  \psfrag{Nt is 100 uniform jamming}[l][l][0.7]{$N_t = 100$, Uniform Jamming}
  \psfrag{Nt is 100 directional jamming}[l][l][0.7]{$N_t = 100$, Directional Jamming}
  \psfrag{Nt is 50 uniform jamming}[l][l][0.7]{$N_t = 50$, Uniform Jamming}
  \includegraphics[width=1\columnwidth,height = 2.5in]{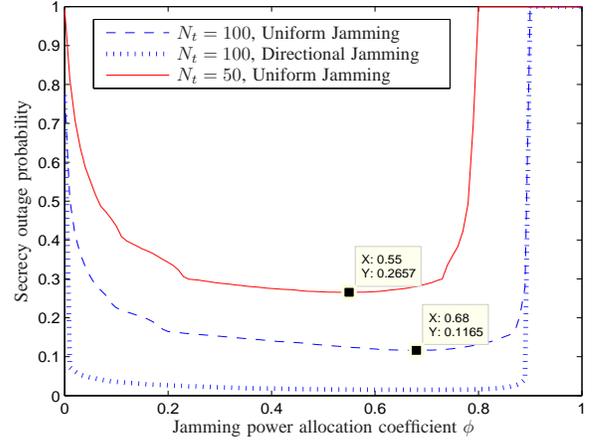}
  \caption{SOP vs. power allocation coefficient $\phi$. Parameters setting: $d_{\rm b} = 100m$, $\theta_{\rm b} = 0^\circ$, $\mathcal{D}_{\rm e} = [50m, 100m]$, $\mathcal{A}_{\rm e} = [-15^\circ, 15^\circ]$, $L$ = 10, and $R_{\rm th} = 10$bps/Hz.}
  \label{fig:1}
\end{figure}

Next, we depict the SOP vs. $\phi$ in Fig.~\ref{fig:1}, with randomly generated $\theta_{{\rm e}}$ and $d_{{\rm e}}$ within the range $\mathcal{A}_{\rm e} = [-15^\circ, 15^\circ]$ and $\mathcal{D}_{\rm e} = [50m, 100m]$, respectively. Note that the smoothless of the curves is not due to the lack of simulation trials, but caused by the fact that $\mathcal{P}_{\rm out}$ is a piecewise function, as shown in \eqref{eq:19} and \eqref{eq:44}. In addition, the SOP rapidly increases to 1 when $\phi$ exceeds $\phi^{\rm max}$ in \eqref{eq:22.1}. From Fig.~\ref{fig:1}, we can see that 1) the SOP with optimal $\phi$ is much smaller than that without jamming, i.e., $\phi = 0$, as anticipated in Proposition \ref{prop:benefit range}; 2) the SOP decreases as $N_t$ increases since larger $N_t$ results in higher received power at Bob and less leakage to Eve. Moreover, the optimal $\phi$ increases with $N_t$ because with larger $N_t$, allocating more power to data transmission is not efficient in increasing the achievable rate of Bob because of the logarithmic slope of the rate function; and 3) the SOP is further substantially reduced with directional jamming.

\begin{figure}[tbp]
  \centering
  \psfrag{Normalized crosstalk between Eve and Bob}[][][0.7]{Normalized crosstalk between Eve and Bob ($s_{{\rm e;b}}$)}
  \psfrag{Optimal jamming power allocation coefficient}[][][0.7]{Optimal jamming power allocation coefficient ($\phi^{\tt opt}$)}
  \psfrag{Monte Carlo}[l][l][0.7]{Monte Carlo}
  \psfrag{Analytical}[l][l][0.7]{Analytical}
  \psfrag{d0 dominates}[][][0.7]{$\phi_0$ dominates}
  \psfrag{dopt dominates}[][][0.7]{$\phi_g^{\tt opt}$ dominates}
  \psfrag{Nt = 100}[][][0.65]{$N_t = 100$}
  \psfrag{db = 100}[][][0.65]{$d_{\rm b} = 100{\rm m}$}
  \psfrag{db = 150}[][][0.65]{$d_{\rm b} = 150{\rm m}$}
  \psfrag{Nt = 50}[][][0.65]{$N_t = 50$}
  \includegraphics[width=1\columnwidth,height = 2.5in]{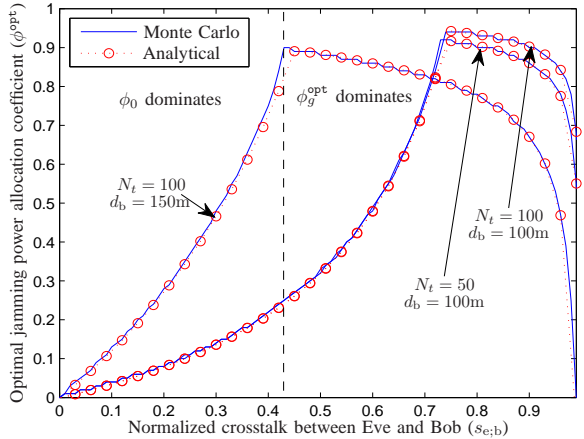}
  \caption{Optimal jamming power coefficient $\phi^{\tt opt}$ vs. $s_{\rm e;b}$ under Corollary 3 with $\mathcal{D}_{\rm e} = [50m, 350m]$ and $R_{\rm th} = 5$ bps/Hz.}
  \label{fig:3}
\end{figure}

Fig.~\ref{fig:3} shows the optimal jamming power coefficient $\phi^{\tt opt}$ as a function of the normalized crosstalk $s_{\rm e;b}$, for uniform jamming. We compare the derived $\phi^{\tt opt}$ in \eqref{eq:phi opt} with Monte Carlo simulations and show a good match between them. From Fig.~\ref{fig:3}, we first observe that each curve is divided into two parts, respectively representing that $\phi_0$ or $\phi_g^{\tt opt}$ dominates the optimal result in \eqref{eq:phi opt}. The division is emphasized using a vertical dash line for the case with $d_{\rm b} = 150$m. When $d_{\rm b} = 100$m, in the $\phi_0$-dominating region, the curves with $N_t = 50$ and $100$m coincide with each other since $N_t$ does not affect $\phi_0$ in \eqref{eq:35.1}. In the $\phi_g^{\tt opt}$-dominating region, $\phi^{\tt opt}$ increases with $N_t$ with given $s_{\rm e;b}$.
The region-division in Fig.~\ref{fig:3} can be explained as follows: in the left region, $s_{\rm e;b}$ is small enough so the channels between Eves and Bob can be considered as asymptotically orthogonal, i.e., $\mathbf{h}_{\rm e} \in {\rm null}(\mathbf{h}_{\rm b})$. In this case, more signal power is leaked to Eve with larger $s_{\rm e;b}$, hence we need more jamming power allocated in ${\rm null}(\mathbf{h}_{\rm b})$ to degrade Eve's channel.
However, in the $\phi_g^{\tt opt}$-dominating region, the value of $s_{\rm e;b}$ is large, which indicates that channels from Alice to Eves and Bob could be highly aligned, i.e., $\mathbf{h}_{\rm e} \notin {\rm null}(\mathbf{h}_{\rm b})$. In this case, jamming in ${\rm null}(\mathbf{h}_{\rm b})$ is not efficient to degrade Eve's channel and the jamming can be a waste of transmit power. Thus, the optimal jamming power starts decreasing with $s_{\rm e;b}$.

\begin{figure}[tbp]
  \centering
  \psfrag{x}[][][0.7]{$d_{\rm b}$ (m)}
  \psfrag{y}[][][0.7]{Area of the secrecy outage region ($m^2$)}
  \psfrag{a}[l][l][0.7]{Without jamming}
  \psfrag{b}[l][l][0.7]{Uniform jamming}
  \psfrag{c}[l][l][0.7]{Directional jamming (Algo. $2$)}
  \psfrag{Directional jamming Algorithm 2 ----}[l][l][0.7]{Directional jamming (Algo. $3$)}
  \includegraphics[width=1\columnwidth,height = 2.5in]{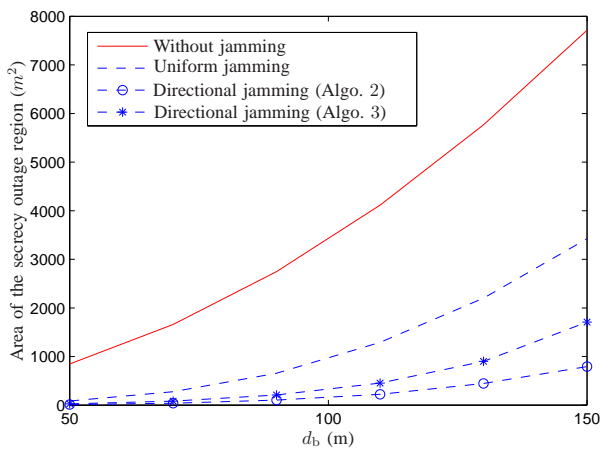}
  \caption{Area of the secrecy outage region vs. $d_{\rm b}$, results are shown for variant jamming algorithms. Parameters setting: $N_t = 50$, $R_{\rm th} = 5$, and $\theta_{\rm b} = 0^\circ$.}
  \label{fig:SORvsd}
\end{figure}

In Fig.~\ref{fig:SORvsd}, for the case that without information of $\mathcal{R}_{\rm sus}$, we compare the area of SOR achieved by different algorithms. For uniform jamming, the optimal $\phi$ is found via one dimensional linear search by minimizing the area described in \eqref{eq:area sor 1}. For all curves, the SOR enlarges with increasing $d_{\rm b}$, since larger $d_{\rm b}$ results in weaker signal power received at Bob. With uniform jamming, the area of the SOR can be reduced approximately by half compared with the non-jamming case. This considerable reduction of SOR area together with the least implementation complexity make uniform jamming still a good option for practical system design. Using the directional jamming with Algorithm $2$, the SOR area can be further reduced, but it induces the highest complexity among all schemes. At last, we note that for the directional jamming, Algorithm $3$ achieves slightly larger area of SOR than Algorithm 2, where the difference becomes smaller especially for small $d_{\rm b}$. However, the implementation complexity can be greatly reduced by Algorithm 3, which makes it being a reasonable choice that strikes a compromise between complexity and performance. It can also be confirmed from Fig.~\ref{fig:SORvsd} that even without any information of $\mathcal{R}_{\rm sus}$, directional jamming can still be utilized to enhance communication secrecy.

\begin{figure}[tbp]
  \centering
  \psfrag{Without jamming}[l][l][0.7]{Without jamming}
  \psfrag{Uniform jamming}[l][l][0.7]{Uniform jamming}
  \psfrag{Directional jamming, Algo. 1 aaaaaaaaaaaaa}[l][l][0.7]{Directional jamming (Algo. 1)}
  \psfrag{Directional jamming, Algo. 3 aaaaaaaaaaaaa}[l][l][0.7]{Directional jamming (Algo. 3)}
  \psfrag{Secrecy outage probability}[l][l][0.7]{Secrecy outage probability}
  \psfrag{db (m)}[l][l][0.7]{$d_{\rm b}$ (m)}
  \includegraphics[width=1\columnwidth,height = 2.5in]{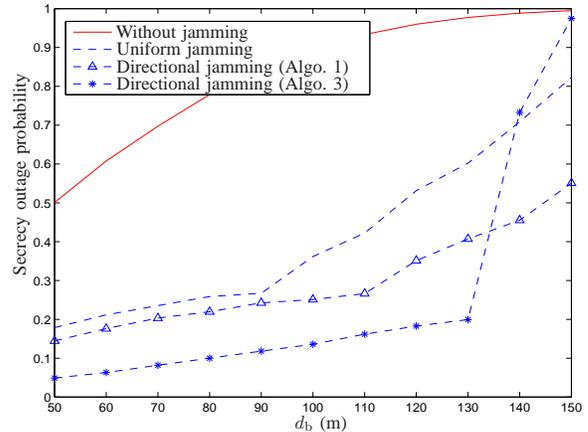}
  \caption{Secrecy outage probability vs. $d_{\rm b}$, results are shown for no jamming and uniform/directional jamming algorithms. Parameters setting: $N_t = 100$, $R_{\rm th} = 10$bps/Hz, $\theta_{\rm b} = 0^\circ$, $L = 10$, $\mathcal{A}_{\rm e} = [-30^\circ, 30^\circ]$, and $\mathcal{D}_{\rm e} = [50{\rm m}, 200{\rm m}]$.}
  \label{fig:SOP}
\end{figure}

{We depict the SOP vs. $d_{\rm b}$ in Fig.~\ref{fig:SOP} for a particular example scenario, where $\mathcal{R}_{\rm sus}$ is defined by $\mathcal{A}_{\rm e} = [-30^\circ, 30^\circ]$ and $\mathcal{D}_{\rm e} = [50{\rm m}, 200{\rm m}]$. Four schemes are compared, i.e., without jamming, uniform jamming, and directional jamming schemes Algorithm 1 and Algorithm 3. Correspondingly, we also depict the SORs achieved by each schemes in Fig.~\ref{fig:SORexample} to help us better understanding the relation between the SOR and SOP. We note that in Fig.~\ref{fig:SOP}, the smoothless of the curves is caused by the area calculation of the intersection between two complicated-shaped regions, not by the lacking of simulation trials.

As shown in Fig.~\ref{fig:SOP}, all jamming schemes can achieve lower SOP compared to that without jamming. In the small-$d_{\rm b}$ region, directional jamming schemes outperform uniform jamming. However, the performance improvement from uniform jamming to directional jamming Algorithm $1$ is small because Algorithm $1$ directionally allocates jamming power towards $\mathcal{R}_{\rm sus}$ to suppress the SOR side lobes (as shown in Fig.~\ref{fig:SORexample} (c)), whereas when Bob is close to Alice, most side lobes are already very small with uniform jamming. Moreover, in the small-$d_{\rm b}$ region, directional jamming Algorithm $3$ achieves the lowest SOP for the reason that it focuses on only two dominating SOR side lobes (as shown in Fig.~\ref{fig:SORexample} (d)) hence the jamming power can be used more efficiently in Algorithm 3 than Algorithm 1, where jamming is uniformly performed to all directions within $\mathcal{R}_{\rm sus}$. In this example, the two dominating side lobes are covered by $\mathcal{R}_{\rm sus}$, thereby they contribute more in the SOP calculation compared with the other side lobes. However, in a different scenario where $\mathcal{R}_{\rm sus}$ does not cover the two dominating side lobes, it cannot be concluded that Algorithm 3 always outperforms Algorithm 1.

At last, we note that as $d_{\rm b}$ increases, the performance of Algorithm 3 degrades and Algorithm 1 outperforms the others. The reason is, when $d_{\rm b}$ is large, all SOR side lobes correspondingly become large as shown in Fig.~\ref{fig:SORexample} (d). In this case, besides the two dominating side lobes, the impact from the other side lobes cannot be simply ignored as that has been done in Algorithm $3$. In conclusion, in practice, appropriate jamming scheme should be determined according to both the available information of $\mathcal{R}_{\rm sus}$ and the location information of Bob such as $d_{\rm b}$.}

\begin{figure}[tbp]
  \centering
\psfrag{Without jamming}[c][c][0.7]{(a) Without jamming}
\psfrag{Uniform jamming, without info. of Rsus}[c][c][0.7]{(b) Uniform jamming}
\psfrag{Directional jamming, with info. of Rsus}[c][c][0.7]{(c) Directional jamming (Algo. 1)}
\psfrag{Directional jamming, without info. of Rsus}[c][c][0.7]{(d) Directional jamming (Algo. 3)}
  \includegraphics[width=1\columnwidth]{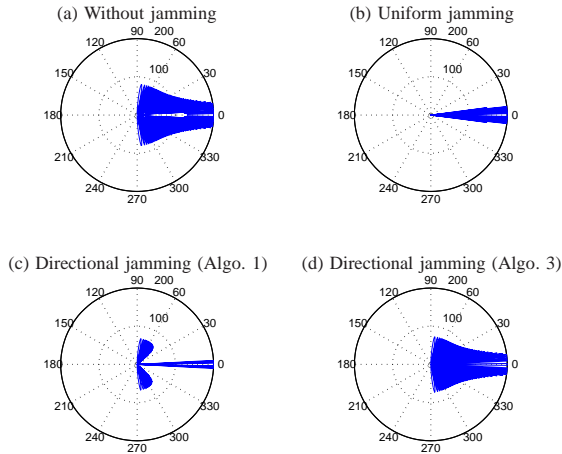}
  \caption{Description of the corresponding SORs for the $4$ schemes in Fig.~\ref{fig:SOP}, when $d_{\rm b} = 150{\rm m}$.}
  \label{fig:SORexample}
\end{figure}

\section{Conclusions}
In this paper, we first define and analytically describe the SOR for secure communication in massive MIMO Rician channels, and derive expressions of the SOP. We then determine a jamming-beneficial range, indicating that uniform jamming is useful in reducing the SOP when the distance from Eve to Alice is less than a threshold. Optimal jamming and signal power allocation is investigated for uniform jamming, furthermore, for both conditions with and without the information of the suspicious area, we propose directional jamming algorithms, which makes use of the jamming power more efficiently and further reduce the SOP. We further extend the discussions to multiuser and multi-cell scenarios where future challenges are also described. In conclusion, we claim that uniform jamming still helps the communication secrecy in massive MIMO systems, and the proposed directional jamming outperforms conventional uniform jamming schemes.

\appendices

\section{Proof of Lemma \ref{le:CT cha}}\label{proof:le CT}
{\em Proof of 1):}
We can rewrite $s_{i;j}(\theta_i, \theta_j)$ in \eqref{eq:23.2} as a function of  $\Delta_{i;j}$ as follows
\begin{equation}\label{eq:13}
  s_{i;j}(\theta_i,\theta_j) = K_{i;j} s\left(\Delta_{i;j}\right).
\end{equation}
By applying \cite[(14)]{RavGlobecom07} to $t_{i;j}$ in \eqref{eq:23.2}, $s(x)$ in \eqref{eq:13} can be represented as
\begin{equation}\label{eq:sx def}
  s\left(x\right) = \left\{ {\begin{array}{*{20}{c}}
{1,}&{x  = 0}\\
{\frac{1}{N_t^2}\frac{{{{\sin }^2}\left( {{N_t}\pi dx } \right)}}{{{{\sin }^2}\left( {\pi dx } \right)}},}&{x  \ne 0}
\end{array}} \right.
\end{equation}
which is a sinc-like function, has one main lobe and side lobes with decreasing amplitudes.

{\em Proof of 2):}
In order to describe the CDF of $s_{i;j}$, we first characterize $s(x)$ in \eqref{eq:sx def} by some cross points and peak values, shown in Fig.~\ref{fig:showpoints} and defined in the following.
\begin{figure}[tbp]
  \centering
  \psfrag{main lobe}[][][0.7]{main lobe}
  \psfrag{s(x)}[][][0.7]{$s(x)$}
  \psfrag{x}[][][0.7]{$x$}
  \psfrag{X}[][][0.7]{$\mathcal{X}_{\mathcal{A}_{i}}$}
  \psfrag{a}[][][0.65]{${\rm CP}_0(u)$}
  \psfrag{b}[][][0.65]{${\rm CP}_{1,1}(u)$}
  \psfrag{c}[][][0.65]{${\rm CP}_{1,2}(u)$}
  \psfrag{d}[][][0.7]{$\beta_{\mathcal{A}_{i}}^*$}
  \psfrag{e}[][][0.7]{${\rm PV}_1$}
  \psfrag{first legend}[l][l][0.7]{$y = s(x)$}
  \psfrag{sec legend}[l][l][0.7]{$y = u$}
  \includegraphics[width=1\columnwidth]{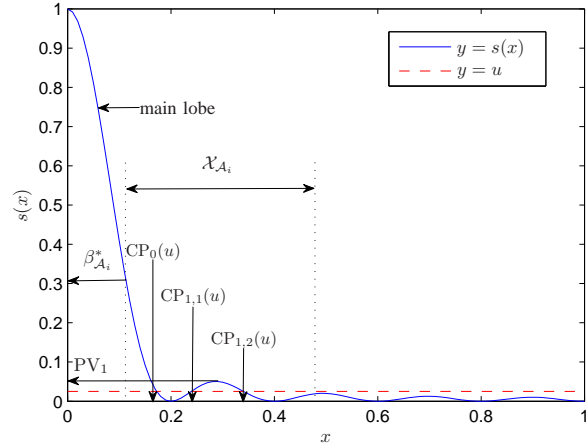}
  \caption{Description of $s(x)$, the cross points and peak values.}
  \label{fig:showpoints}
\end{figure}
\begin{definition}
The cross points and peak values are defined as
  \begin{itemize}
  \item {\em Peak Values}: The peak value of the $m$-th side lobe of $s(x)$ can be approximated by
  \begin{equation}\label{eq:peak value}
    {\rm PV}_m \approx \frac{1}{N_t^2}\frac{1}{\sin^2\left(\pi\frac{m + \frac{1}{2}}{N_t}\right)}\approx \frac{1}{\pi^2 \left(m+\frac{1}{2}\right)^2}
  \end{equation}
  which {is obtained by noting that $\sin (x) \approx x$ when $x$ is small,} and it becomes asymptotically exact as $N_t$ is large. ${\rm PV}_0 = 1$ corresponds to the main lobe.
  \item {\em Cross Points}: When $u < {\rm PV}_M$, ${\rm CP}_{m,i}(u), i=1,2$ denotes the $i$-th cross point between $y = u$ and $y = s(x)$ in side lobe $m$ ($m<M$), ${\rm CP}_0(u)$ is the cross point in the main lobe.
\end{itemize}
\end{definition}

 Using \eqref{eq:13}, we have $F_{s_{i;j}}(K_{i;j}u) = \Pr\{s\left(\Delta_{i;j}\right) < u\}$. According to Fig.~\ref{fig:showpoints}, $\Pr\{s\left(\Delta_{i;j}\right) < u\}$ can be evaluated by calculating the probability that $\Delta_{i;j}$ falls within the discrete intervals determined by ${\rm CP}_{m,i}(u)$ and ${\rm CP}_0(u)$. On the other hand, recalling that $\Delta_{i;j} \triangleq |\sin \theta_{i} - \sin \theta_{j}|$ and $\theta_{i}$ follows uniform distribution, the CDF of $\Delta_{i;j}$ can be written as
  \begin{multline}\label{eq:46.1}
  F_{\Delta}(z)\\
   \triangleq \frac{1}{\theta_{\max} - \theta_{\min}}\bigg[ \min\left(\sin^{-1}\left(\min\left(1, z + \sin\theta_{j}\right)\right), \theta_{\max}\right) \\
   -\max\left(\sin^{-1}\left(\max\left(-1, -z + \sin\theta_{j}\right)\right), \theta_{\min}\right)\bigg]^+.
\end{multline}
 Using $F_{\Delta}(z)$ to describe the probability that $\Delta_{i;j}$ falls within the described intervals leads to \eqref{eq:19}.

{\em Proof of 3):}
Given $\theta_{j}$, the feasible range of $\Delta_{i;j}$ is determined by $\mathcal{A}_i$, i.e., the angle range that node $i$ distributed in. Use $\mathcal{X}_{\mathcal{A}_{i}}$ to denote this feasible range (shown in Fig.~\ref{fig:showpoints}), it holds that $\mathcal{X}_{\mathcal{A}_{i}} \subset [0,1]$.
Given $\mathcal{A}_{i}$, the feasible range of $s_{i;j}$ can be determined as $[0, K_{i;j}\beta^*_{\mathcal{A}_{i}}]$, where $\beta^*_{\mathcal{A}_{i}} \triangleq \mathop {\max}\limits_{x \in \mathcal{X}_{\mathcal{A}_{i}}} s(x)$. It is clear that $\beta^*_{\mathcal{A}_{i}} \leq 1$, leading us to the conclusion.

\section{Proof of Proposition \ref{prop:benefit range}}\label{proof:benificial}
For the ease of description, we first define
  \begin{equation}\label{eq:22}
 h\left(x\right) \triangleq \frac{x N_t 2^{R_{\rm th}}}{{1 + x d^{-\alpha}_b N_t} - 2^{R_{\rm th}}}.
  \end{equation}
Letting $t = \frac{z^{\alpha} + \tilde{P}_{\rm jam}}{a_1 + \tilde{P}_{\rm jam}}$ where $a_1 =h (\tilde{P}_b )$, and $t_0 = \frac{z^{\alpha}}{a_2}$ where $a_2 = h (\tilde{P}_{\rm tot})$, we rewrite \eqref{eq:44} as
\begin{equation}
  \mathcal{P}_{\rm out} = 1 -
    \Bigg\{\int_{d_{\min}}^{d_{\max}} F_{s_{{\rm e};{\rm b}}}\left(t\right)
    \times \frac{2z}{d_{\max}^2 - d_{\min}^2} {\rm d}z\Bigg\}^L.
\end{equation}
{Since $F_{s_{{\rm e};{\rm b}}}(\cdot)$ is an increasing function}, $\mathcal{P}_{\rm out}$ decreases with $t$ in the domain of $F_{s_{{\rm e};{\rm b}}}(\cdot)$. Therefore, if the following conditions are satisfied, we can conclude that jamming is beneficial.
\begin{enumerate}
\item
There exists a positive value of $\tilde{P}_{\rm jam}$, which holds $t - t_0 > 0$ for any $z \in \mathcal{D}_{\rm e}$; and
\item
$z \in \mathcal{D}_e$ such that $t_0 < s_{\rm e;b}^{\max}$
 \end{enumerate}
where {condition 1) defines the scenario, in which jamming can always result in a larger $t$ than $t_0$ (corresponding to the no-jamming case), thereby leading to lower $\mathcal{P}_{\rm out}$ since it is a decreasing function in $t$. Moreover, condition 2) ensures that $t_0$ is less than the maximum feasible value of $s_{\rm e;b}$, otherwise, $\mathcal{P}_{\rm out}$ is always zero either with or without jamming under condition 1), hence the benefits of jamming cannot be concluded.}
For 1), it is equivalent to prove that
  \begin{equation}\label{eq:27.2}
    \frac{z^{\alpha} + \tilde{P}_{\rm jam}}{a_1 + \tilde{P}_{\rm jam}} > \frac{z^{\alpha}}{a_2}
  \end{equation}
has a positive solution of $\tilde{P}_{\rm jam}$.
Note that with proper parameter setting that $\phi$ is no larger than $\phi^{\max}$ described in \eqref{eq:22.1}, we have $a_1 > 0$ and $a_2 > 0$, and note that if
  \begin{equation}\label{eq:28.1}
    d_{\max}^{\alpha} < a_2
  \end{equation}
  then $z^{\alpha} < a_2$ holds for any $z$. In this case, \eqref{eq:27.2} can be equivalently rewritten as
  \begin{equation}\label{eq:29.1}
    \tilde{P}_{\rm jam} > \frac{(a_1 - a_2)z^{\alpha}}{a_2 - z^{\alpha}}.
  \end{equation}
Noting that $\frac{{\rm d}h(x)}{{\rm d}x} < 0$, it is clear from \eqref{eq:29.1} that $\tilde{P}_{\rm jam} > 0$, thus \eqref{eq:28.1} is a sufficient condition to satisfy 1).
On the other hand, as $t_0 \leq \frac{d_{\max}^{\alpha}}{a_2}$, a sufficient condition that satisfies 2) is
\begin{equation}\label{eq:31.2}
  {d_{\max}^{\alpha}} < s_{\rm e;b}^{\max}a_2.
\end{equation}
 Here, according to Lemma \ref{le:CT cha}, we have $s_{\rm e;b}^{\max} \leq 1$. Hence, the intersection of \eqref{eq:28.1} and \eqref{eq:31.2} is equal to \eqref{eq:31.2}, and recalling the definition of $a_2$ leads to the final result.

\section{Proof of Corollary \ref{coro:area side lobe}}\label{proof:coro2}

First, from \eqref{eq:13}, \eqref{eq:sx def}, $s_{\rm e;b}(\theta_{\rm e})$ in the $m$-th side lobe can be presented as
\begin{equation}\label{eq:50}
  s_{\rm e;b}(\theta_{\rm e}) = \frac{K_{\rm e;b}}{N_t^2} \frac{{{{\sin }^2}\left( {{N_t}\pi dx } \right)}}{{{{\sin }^2}\left( {\pi dx } \right)}},~ x\in\left[\frac{m}{N_t d}, \frac{m + 1}{N_t d}\right].
\end{equation}
In \eqref{eq:50}, sine functions appear in both the numerator and denominator, making it infeasible to obtain closed-form results in further analysis. Hence, we find an upper bound of \eqref{eq:50} by fixing the value of $x$ to the left end point of its domain, i.e., $x = \frac{m}{N_t d}$, in the denominator.
\begin{equation}\label{eq:50.1}
  s_{\rm e;b}(\theta_{\rm e}) \leq \frac{K_{\rm e;b}}{\pi^2 m^2}\sin^2\left(N_t \pi d x\right).
\end{equation}
For the ease of description, we consider only the condition that $\theta_{\rm e} > 0$. With $\theta_{\rm b} = 0$, we can rewrite $x$ (defined as $x = |\sin \theta_{\rm e} - \sin \theta_{\rm b}|$) in terms of $\theta_{\rm e}$ as
\begin{equation}\label{eq:x theta}
x = \sin\theta_{\rm e} \leq \theta_{\rm e}.
\end{equation}
Note that the side lobes that are close to the main lobe are more important in contributing to the area of the SOR, as a result, the value of $\theta_{\rm e}$ that should be concerned is very small. Therefore, the upper bound in \eqref{eq:x theta} can be very tight. Combining \eqref{eq:50.1} and \eqref{eq:x theta} we have
\begin{equation}\label{eq:53}
  s_{\rm e;b}(\theta_{\rm e}) \leq \frac{K_{\rm e;b}}{\pi^2 m^2}\sin^2\left(N_t \pi d  \theta_{\rm e}\right).
\end{equation}
Substitute \eqref{eq:53} into \eqref{eq:39.1}, and calculate the area by integral, we get
\begin{align}\label{eq:54}
&{\rm Area}\left(\mathcal{R}^{\rm uj}_{{\rm SOR},m}\right) \nonumber\\
&= \int_{\theta_{\rm e}\in \mathcal{A}_{m}} \frac{1}{2}\left(C_1(\phi)s_{\rm e;b}(\theta_{\rm e}) - C_2(\phi)\right) {\rm d}\theta_{\rm e}\\
\label{eq:55}
&\leq  \int_{\theta_{\rm e}\in \mathcal{A}_{m}} \frac{1}{2}\left(C_1(\phi)\frac{K_{\rm e;b}}{\pi^2 m^2}\sin^2\left(N_t \pi d  \theta_{\rm e}\right) - C_2(\phi)\right) {\rm d}\theta_{\rm e}\\
\label{eq:56}
&= \frac{1}{N_t \pi d}\int_{m\pi}^{(m + 1)\pi} \frac{1}{2}\left(C_1(\phi)\frac{K_{\rm e;b}}{\pi^2 m^2}\sin^2\left(\eta\right) - C_2(\phi)\right) {\rm d}\eta\\
\label{eq:57}
&= \frac{1}{4N_t d}\left(\frac{K_{\rm e;b}}{\pi^2 m^2}C_1(\phi) - 2C_2(\phi)\right)
\end{align}
where $\mathcal{A}_{m} = [\frac{m}{N_t d}, \frac{m+1}{N_t d}]$ is the physical angle range of the $m$-th side lobe.
To obtain \eqref{eq:54}, we use $\alpha = 2$, and \eqref{eq:55} is obtained using \eqref{eq:53}. From \eqref{eq:55} to \eqref{eq:56}, we use the variable substitution $\eta = N_t \pi d  \theta_{\rm e}$, then \eqref{eq:57} is obtained by simple integral calculation of elementary functions.

\section{Proof of Corollary \ref{coro:equal theta}}\label{proof:coroET}
When $\theta_{{\rm e}_l} = \theta_{\rm e}$, $\forall l = 1,...,L$, $\mathcal{P}_{\rm out}$ in \eqref{eq:44} is determined only by $d_{\rm e}$ such that
\begin{equation}\label{eq:42.1}
  \mathcal{P}_{\rm out} = 1 - \left(\int_{\max[d_0(\phi,\theta_{\rm e}), d_{\min}]}^{d_{\rm max}}\frac{2z}{d_{\max}^2 - d_{\min}^2}{\rm d}z\right)^L
\end{equation}
where
\begin{align}\label{eq:34.1}
  d_0(\phi,\theta_{\rm e}) &= \left[\left\{h\left((1 - \phi)\tilde{P}_{\rm tot}\right)+ \phi \tilde{P}_{\rm tot}\right\}s_{{\rm e};{\rm b}}(\theta_{\rm e}) - \phi \tilde{P}_{\rm tot}\right]^{\frac{1}{\alpha}}\nonumber\\
  &\triangleq [g(\phi)]^{\frac{1}{\alpha}}
\end{align}
is the distance threshold. {Given $\theta_{\rm e}$, any Eve with a distance to Alice smaller than $d_0(\phi,\theta_{\rm e})$ can cause secrecy outage. Noting this, \eqref{eq:42.1} calculates the overall outage probability by assuming that $d_{\rm e}$ follows uniform distribution.}
Clearly, if $\mathop {\min }\limits_\phi  {d_0(\phi,\theta_{\rm e})} < d_{\min}$, {the integral range in \eqref{eq:42.1} becomes $[d_{\min},d_{\max}]$ hence }a zero SOP can be achieved; otherwise {if $\mathop {\min }\limits_\phi  {d_0(\phi,\theta_{\rm e})} > d_{\max}$, }a definite outage occurs with probability $1$.\footnote{To facilitate the analysis, we assume the definite outage will not happen by letting $d_{\max}$ to be large enough such that $d_{\max}> \mathop {\max }\limits_{\theta_{\rm e}} \mathop {\min }\limits_\phi  {d_0(\phi,\theta_{\rm e})}$.}  Rewrite $g(\phi)$ in \eqref{eq:34.1} as
\begin{equation}\label{eq:34}
  g(\phi) \triangleq s_{\rm e;b}(\theta_{\rm e})g_1(\phi) + (1 - s_{\rm e;b}(\theta_{\rm e}))g_2(\phi)
\end{equation}
where $g_1(\phi) \triangleq h((1 - \phi)\tilde{P}_{\rm tot})$, $g_2(\phi) \triangleq -\phi \tilde{P}_{\rm tot}$ and $h(\cdot)$ is defined in \eqref{eq:22}. Clearly, minimizing \eqref{eq:42.1} is equivalent to minimizing \eqref{eq:34}.
Since $g_1(\phi)$ is concave and $g_2(\phi)$ is linear, $g(\phi)$ is concave. Hence, letting $\frac{\partial g(\phi )}{\partial \phi } = 0$, the optimal $\phi$ satisfying $\phi < \phi^{\max}$ (as in \eqref{eq:22.1}) is given by $\phi^{\tt opt}_g$ in
\eqref{eq:35}.

Note that if $d_0(\phi^{\tt opt}_g, \theta_{\rm e}) < d_{\min}$, setting $\phi^{\tt opt} = \phi^{\tt opt}_g$ is not the only choice since the solution of $d_0(\phi,\theta_{\rm e}) = d_{\min}$ (if exists), denoted as $\phi_0$, also achieves zero SOP.
Therefore, we need to check the value of $\phi_0$ and compare it with $\phi^{\tt opt}_g$ to determine the final optimal jamming power allocation. Note that in \eqref{eq:34}, the value of $g_1(\phi)$ is usually much less than $g_2(\phi)$. Hence, when $s_{\rm e;b}(\theta_{\rm e})$ is not so large, $g(\phi)$ can be well approximated by a linear function as $g(\phi) \approx s_{\rm e;b}(\theta_{\rm e})h({\tilde P}_{\rm tot}) - (1 - s_{\rm e;b}(\theta_{\rm e})){\tilde P}_{\rm tot}\phi$. Using this approximation, we get
$\phi_0 = \frac{s_{\rm e;b}(\theta_{\rm e})h({\tilde P}_{\rm tot}) - d_{\min}^{\alpha}}{(1 - s_{\rm e;b}(\theta_{\rm e})){\tilde P}_{\rm tot}}$.

If $\phi_0  \notin  [0,1]$, it is not a feasible solution in practice. If $\phi_0  \in  [0,1]$, both $\phi_0$ and $\phi_g^{\tt opt}$ are able to achieve zero SOP. In this case, we choose the smaller one between $\phi^{\tt opt}_g$ and $\phi_0$ to save jamming power such that
$\phi^{\tt opt} = \min\{\phi^{\tt opt}_g, \phi_0\}$.

\section{Derivation of Algorithm 3}\label{proof:alg3}
Algorithm 3 is proposed based on the following assumptions and approximations:


\begin{enumerate}
  \item Assuming that the angle ranges occupied by every side lobe (and the main lobe) of $\mathcal{R}_{\rm SOR}^{\rm dj}(\bar{\mathbf{p}})$ are approximately the same (which is $\frac{\pi}{M+1}$ assuming that there are in total one main lobe and $M$ side lobes within the angle range $\left[\frac{\pi}{2},\frac{\pi}{2}\right]$), an upper bound of the integral in \eqref{eq:71} can be approximated by
      \begin{equation}\label{eq:72}
        {{\rm Area}_{\rm UB}\left(\mathcal{R}_{\rm SOR}^{\rm dj}(\bar{\mathbf{p}})\right)} \approx \frac{\pi}{2(M+1)} \sum\limits_{m = 0}^M \left(\bar{d}_{{\rm e}}^{\rm dj}(\bar{\mathbf{p}}, \theta_m)\right)^2
      \end{equation}
      {where the area of every side lobe is upper bounded by the area of its enclosing sector.}
  \item We assume that $\bar{\mathbf{s}}({\theta_{m}})^H \mathbf{v}_m = 1$ and $\bar{\mathbf{s}}({\theta_{n}})^H \mathbf{v}_m = 0, \forall n\neq m$. In practice, by defining $\mathbf{v}_m = \frac{\bar{\mathbf{s}}({\theta_{m}})}{||\bar{\mathbf{s}}({\theta_{m}})||}$, this assumption is easy to realize with large $N_t$, where asymptotic orthogonality holds.
      Now, according to \eqref{eq:70}, $\bar{d}_{{\rm e}}^{\rm dj}(\bar{\mathbf{p}}, \theta_m)$ in \eqref{eq:72} can be written as
      \begin{equation}\label{eq:73}
        \bar{d}_{{\rm e}}^{\rm dj}(\bar{\mathbf{p}}, \theta_m) = \left[\left(a_m - \tilde{P}_m\right)^{\frac{1}{\alpha}}\right]^+
      \end{equation}
      where $a_m = \frac{(1 - \phi)\tilde{P}_{\rm tot}N_t 2^{R_{\rm th}} s_{{\rm e;b}}(\theta_{m})}{{1 + (1 - \phi)\tilde{P}_{\rm tot}d^{-\alpha}_b N_t} - 2^{R_{\rm th}}}$ and $\tilde{P}_m = \frac{\bar{P}_m}{N_0}$ is the $m$-th element of $\tilde{\mathbf{p}}$, with $\bar{P}_m$ being the jamming power allocated on the $m$-th side lobe.
\end{enumerate}
From \eqref{eq:72} and \eqref{eq:73}, and given fixed $\phi$, we rewrite the original area minimization problem as
\begin{align}\label{eq:74}
  \min &\sum\limits_{m = 1}^M \left(a_m - \tilde{P}_m\right)^\frac{2}{\alpha}\\
  \label{eq:75}
  {\rm s.t.} &\sum_{m = 0}^M \bar{P}_m = \phi {P}_{\rm tot}
  , 0 \leq \tilde{P}_m \leq a_m.
\end{align}
By checking the Hessian matrix, it is easy to show that the objective function in \eqref{eq:74} is concave. To minimize a concave function, clearly, the optimal solution can be found only on the boundaries of the domain defined by \eqref{eq:75}. Recalling that the area of the side lobes decreases rapidly with larger lobe index, and being aware that jamming should be avoided within the main lobe to prevent degrading Bob's channel, we further simplify the problem by checking the boundary of the domain as described in \eqref{eq:boundary}.
 According to these discussions, Algorithm 3 is obtained.


\vspace{-1cm}

\begin{IEEEbiography}[{\includegraphics[width=1in,height=1.25in,clip,keepaspectratio]{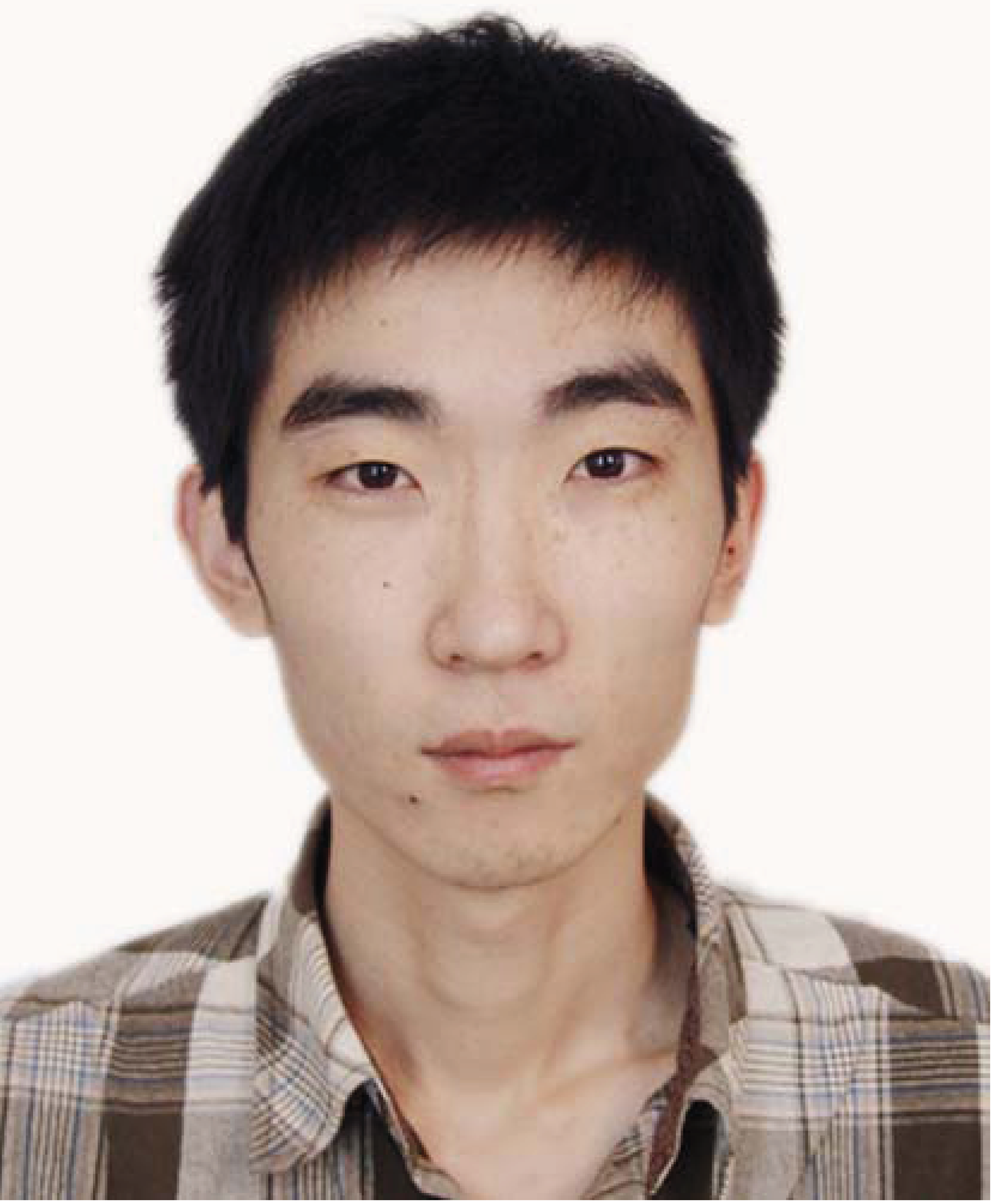}}]{Jue Wang}
(S'10-M'14) received the B.S. degree in communications engineering from Nanjing
University, Nanjing, China, in 2006, the M.S. degree and Ph. D. degree from the National Communications Research Laboratory, Southeast
University, Nanjing, China, respectively in 2009 and 2014.
 In 2014, he joined the School of Electronic and Information Engineering, Nantong University, Nantong, China. Meanwhile, he is with Singapore University of Technology and Design (SUTD) as a post-doctoral research fellow.

Dr. Wang has served as Technical Program Committee member for several IEEE conferences, and reviewer for several IEEE journals. He was awarded as Exemplary Reviewer of {\scshape IEEE Transactions on Communications} for 2014.
His research interests include
MIMO wireless communications, multiuser transmission, MIMO
channel modeling, massive MIMO systems and physical layer security.
\end{IEEEbiography}

\vspace{-1cm}

\begin{IEEEbiography}[{\includegraphics[width=1in,height=1.25in,clip,keepaspectratio]{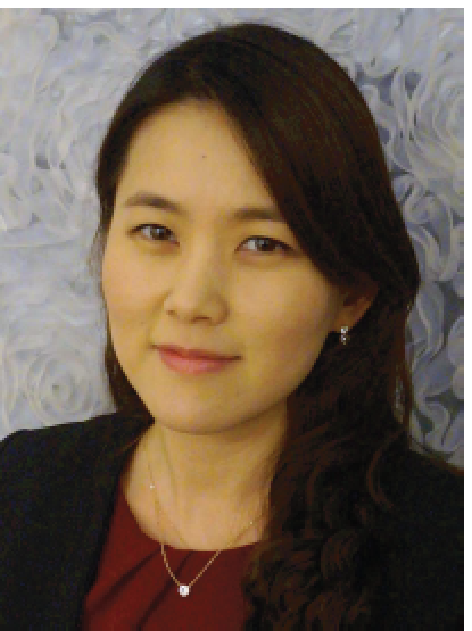}}]
{\bfseries Jemin Lee} (S'06-M'11)  is a Temasek Research Fellow 
at the Singapore University of Technology and Design (SUTD), Singapore.
She received the B.S. (with high honors), M.S., and Ph.D. degrees in Electrical and Electronic Engineering from Yonsei University, Seoul, Korea, in 2004, 2007, and 2010, respectively. She was a Postdoctoral Fellow at the Massachusetts Institute of Technology (MIT), Cambridge, MA from Oct. 2010 to Oct. 2013, and a Visiting Ph.D. Student at the University of Texas at Austin, Austin, TX from Dec. 2008 to Dec. 2009. Her current research interests include physical layer security, wireless security, heterogeneous networks, cognitive radio networks, and cooperative communications.
%

Dr.~Lee is currently an Editor for the {\scshape IEEE Transactions on Wireless Communications} and the {\scshape IEEE Communications Letters},
and served as a Guest Editor of the Special Issue on Heterogeneous and Small Cell Networks for the {\scshape ELSEVIER Physical Communication} in 2014.
She also served as a Co-Chair of the IEEE 2013 Globecom Workshop on Heterogeneous and Small Cell Networks,
and Technical Program Committee Member for numerous IEEE conferences.
She is currently a reviewer for several IEEE journals and has been recognized as an Exemplary Reviewer of {\scshape IEEE Communications Letters} and {\scshape IEEE Wireless Communication Letters} for recent several years.
%
She received the IEEE ComSoc Asia-Pacific Outstanding Young Researcher Award in 2014, the Temasek Research Fellowship in 2013, the Chun-Gang Outstanding Research Award in 2011, and the IEEE WCSP Best Paper Award in 2014.
\end{IEEEbiography}

\begin{IEEEbiography}[{\includegraphics[width=1in,height=1.25in,keepaspectratio]{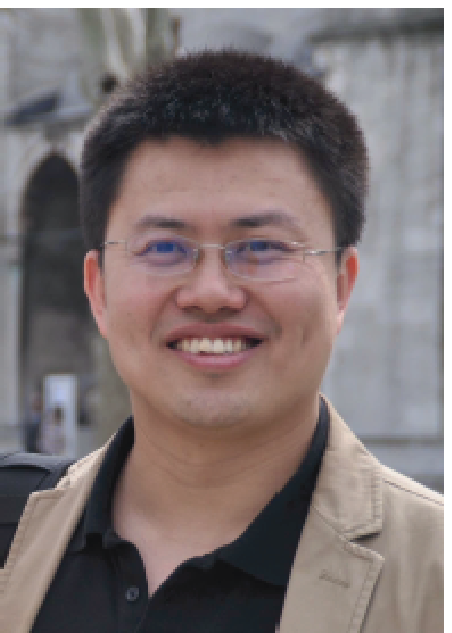}}]
{Fanggang Wang} (S'10-M'11) received the B.Eng. degree in 2005 and the Ph.D. degree in 2010 from the School of Information and Communication Engineering at Beijing University of Posts and Telecommunications, Beijing, China. From 2008 to 2010, he worked as a Visiting Scholar in Electrical Engineering Department, Columbia University, New York City, New York, USA. He was a Postdoctoral Fellow in Institute of Network Coding, the Chinese University of Hong Kong, Hong Kong SAR, China, from 2010 to 2012. He joined the State Key Lab of Rail Traffic Control and Safety, School of Electronic and Information Engineering, Beijing Jiaotong University, in 2010, where he is currently an Associate Professor. His research interests are in wireless communications, signal processing, and information theory. He chaired two workshops on wireless network coding (NRN 2011 and NRN 2012) and served as an Editor in several journals and the Technical Program Committee (TPC) members in several conferences.
\end{IEEEbiography}

\vspace{-7cm}

\begin{IEEEbiography}[{\includegraphics[width=1in,height=1.25in,keepaspectratio]{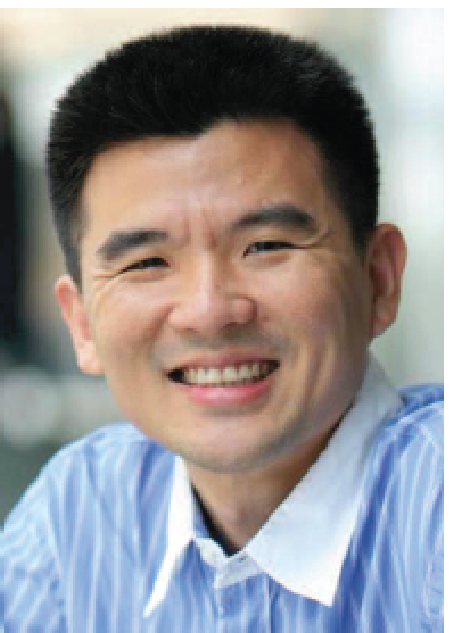}}]
{Tony Q. S. Quek}(S'98-M'08-SM'12) received the B.E.\ and M.E.\ degrees in Electrical and Electronics Engineering from Tokyo Institute of Technology, Tokyo, Japan, respectively. At Massachusetts Institute of Technology, he earned the Ph.D.\ in Electrical Engineering and Computer Science. Currently, he is an Assistant Professor with the Information Systems Technology and Design Pillar at Singapore University of Technology and Design (SUTD). He is also a Scientist with the Institute for Infocomm Research. His main research interests are the application of mathematical, optimization, and statistical theories to communication, networking, signal processing, and resource allocation problems. Specific current research topics include heterogeneous networks, green communications, smart grid, wireless security, internet-of-things, big data processing, and cognitive radio.

Dr.\ Quek has been actively involved in organizing and chairing sessions, and has served as a member of the Technical Program Committee as well as symposium chairs in a number of international conferences. He is serving as the technical chair for the PHY \& Fundamentals Track for IEEE WCNC in 2015, the Communication Theory Symposium for IEEE ICC in 2015, the PHY \& Fundamentals Track for IEEE EuCNC in 2015, and the Communication and Control Theory Symposium for IEEE ICCC in 2015. He is currently an Editor for the {\scshape IEEE Transactions on Communications}, the {\scshape IEEE Wireless Communications Letters}, and an Executive Editorial Committee Member for the {\scshape IEEE Transactions on Wireless Communications}. He was Guest Editor for the {\scshape IEEE Signal Processing Magazine} (Special Issue on Signal Processing for the 5G Revolution) in 2014, and the {\scshape IEEE Wireless Communications Magazine} (Special Issue on Heterogeneous Cloud Radio Access Networks) in 2015.

Dr.\ Quek was honored with the 2008 Philip Yeo Prize for Outstanding Achievement in Research, the IEEE Globecom 2010 Best Paper Award, the CAS Fellowship for Young International Scientists in 2011, the 2012 IEEE William R. Bennett Prize, the IEEE SPAWC 2013 Best Student Paper Award, and the IEEE WCSP 2014 Best Paper Award.
\end{IEEEbiography}

\end{document}